%% file: main.tex
\documentclass{IEEEtaes}
\usepackage{amsmath, amssymb, amsfonts, mathtools}
\usepackage{xargs, tensor, soul, units, cite, stmaryrd, mathrsfs}
\usepackage{algpseudocode, algorithm}
\usepackage{graphicx, xcolor}
\usepackage{pgfplots}
\pgfplotsset{compat=newest}
\usepackage{tikz}
\usetikzlibrary{3d, intersections, arrows.meta, positioning, calc, spy}
\usepackage{multirow}
\usepackage{booktabs}
\usepackage{amsthm}

\tikzset{new spy style/.style={
  spy scope={
    magnification=2,
    size=2cm, 
    connect spies,
    every spy on node/.style={
      rectangle, draw, thick}, 
    every spy in node/.style={
      draw, rectangle, fill=none} 
  }
}}

\usepackage{dblfloatfix}
\usepackage{amsthm}
\theoremstyle{plain}
\theoremstyle{plain}

\newtheorem{assumption}{Assumption}
\newtheorem{definition}{Definition}
\newtheorem{remark}{Remark}
\newtheorem{lemma}{Lemma}
\newtheorem{corollary}{Corollary}
\newtheorem{theorem}{Theorem}

\DeclareMathOperator\proj{\mathtt{proj}}

\DeclareMathOperator{\diag}{diag}

\newcommand\scalemath[2]{\scalebox{#1}{\mbox{\ensuremath{\displaystyle #2}}}}

\definecolor{color1}{RGB}{0,128,0}      
\definecolor{color2}{RGB}{255,0,0}      
\definecolor{color3}{RGB}{0,0,255}      
\definecolor{color4}{RGB}{0,0,0}        
\definecolor{color5}{RGB}{255,140,0}    
\definecolor{color6}{RGB}{75,0,130}     
\definecolor{color7}{RGB}{220,20,60}    
\definecolor{color8}{RGB}{128,128,0}    
\definecolor{color9}{RGB}{0,139,139}    
\definecolor{color10}{RGB}{139,69,19}   
\definecolor{color11}{RGB}{128,0,128}   
\definecolor{color12}{RGB}{105,105,105} 
\definecolor{color13}{RGB}{0,139,139}     

\allowdisplaybreaks
\title{A Switched Systems Approach to Image-Based Feature Tracking for Autonomous Satellite Inspection}

\author{Tochukwu E. Ogri}
\affil{Department of Mechanical and Aerospace Engineering, University of Florida, Gainesville, FL, USA (e-mail: tochukwu.ogri@ufl.edu)}

\author{Muzaffar Qureshi}
\affil{Department of Mechanical and Aerospace Engineering, University of Florida, Gainesville, FL, USA (e-mail: muzaffar.qureshi@ufl.edu)}

\author{Zachary I. Bell}
\affil{Air Force Research Laboratory, Eglin, FL, USA (e-mail: zachary.bell.10@us.af.mil)}

\author{Wanjiku A. Makumi}
\affil{Air Force Research Laboratory, Eglin, FL, USA (e-mail: wanjiku.makumi@us.af.mil)}

\author{Kyle Volle}
\affil{Torch Technologies, Shalimar, FL, USA (e-mail: Kyle.Volle@torchtechnologies.com)}

\author{Rushikesh Kamalapurkar}
\affil{Department of Mechanical and Aerospace Engineering, University of Florida, Gainesville, FL, USA (e-mail: rkamalapurkar@ufl.edu)}

\receiveddate{
This work was supported in part by the Air Force Research Laboratory under Contract Nos. FA8651-24-1-0019 and FA8651-23-1-0006. 
The views expressed are those of the authors and do not necessarily reflect those of the sponsoring agencies.}

\corresp{(Corresponding author: Tochukwu E. Ogri)}

\authoraddress{Tochukwu E. Ogri, Muzaffar Qureshi, and Rushikesh Kamalapurkar are with the Department of Mechanical and Aerospace Engineering, University of Florida, Gainesville, FL, USA. 
Zachary I. Bell and Wanjiku A. Makumi are with the Air Force Research Laboratory, Eglin, FL, USA. 
Kyle Volle is with Torch Technologies, Shalimar, FL, USA.}

\begin{document}
\maketitle
\thispagestyle{empty}
\pagestyle{empty}

\begin{abstract} 
This paper presents an information-based guidance and control architecture for an autonomous deputy spacecraft tasked with inspecting a chief satellite in orbit. The primary objective is for the deputy spacecraft to maximize information gain while tracking features on the chief satellite. The deputy spacecraft needs to respect various constraints such as illumination, field-of-view (FOV), fuel limitations, and avoidance regions. Additionally, the absence of GPS information poses a significant challenge for relative localization within the space environment.
To learn the structure of the chief satellite while achieving relative self-localization and maximizing the information gain, this paper integrates a memory regression extension (MRE)-based distance observer with an information-maximizing adaptive controller. The distance observer utilizes feedback from a camera. A switched systems approach is used to determine the minimum dwell time required for a feature to remain within the FOV of the camera to ensure accurate estimation. A k-means clustering algorithm acts as a high-level planner to intermittently generate goal locations that guide the deputy spacecraft toward the nearest cluster of uninspected points on the chief satellite, subject to illumination and FOV constraints.
A Lyapunov-based stability analysis is conducted to analyze the developed architecture, and simulation results validate the theoretical results of the paper.
\end{abstract}

\section{Introduction}\label{section:introduction}
Guidance and control of autonomous spacecraft carrying out the on-orbit inspection of satellites requires onboard sensors to generate a local map of the environment and also to localize the spacecraft relative to celestial objects in the local map for GPS-denied navigation. For Earth-orbiting spacecraft performing on-orbit servicing, GPS receivers and ground-based radar tracking are commonly used for navigation. In deep space, determining the position of a spacecraft is achieved using the Deep Space Network (DSN) and an onboard radio transponder \cite{SCC.Yost.Weston.ea2024}. Newer approaches under development explore celestial navigation methods that utilize optical detection of planets and X-ray pulsars to determine positional data \cite{SCC.Kapas.Bozoki.ea2021, SCC.Yost.Weston.ea2024}. However, despite significant advancements in algorithms for simultaneous localization and mapping over the last few decades \cite{SCC.Placed.Strader.ea2023}, motion planning for spacecraft guidance and control in complex space environments remains a challenging problem\cite{SCC.Fourie.Tweddle.ea2014, SCC.Nakka.Hoenig.ea2022, SCC.Starek.Schmerling.ea2017, SCC.Volpe.Circi.ea2022,SCC.Maestrini.Di.ea2022, SCC.Wijk.Dunlap.ea2023, SCC.Choi.Nakka.ea2023, SCC.McQuinn.Dunlap.ea2024, SCC.Dunlap.Hamilton.ea2024, SCC.Lin.Wu.ea2024,SCC.Akhihiero.Olawoye.ea2024}.

One approach to achieve relative orbital control during satellite inspection is an information-based guidance and control architecture \cite{SCC.Nakka.Hoenig.ea2022, SCC.Choi.Nakka.ea2023}, where an optimal control problem is formulated using a Bayesian-based information metric (see \cite{SCC.Schwager.Slotine.ea2011, SCC.Schwager.Julia.ea2011}) to maximize the information contained in the images captured by cameras mounted on each deputy spacecraft. While this Bayesian metric works well when the localization of the deputy spacecraft is independent of the camera, it does not readily extend to situations where the same camera must also estimate the pose of the deputy spacecraft relative to the chief satellite. In scenarios where the camera serves both inspection and localization purposes, the information gain is dependent on the trajectory of the camera and therefore cannot be accurately quantified from a single image.

In \cite{SCC.Tweddle.SaenzOtero.ea2015}, a photogrammetry-based scheme is combined with a multiplicative extended Kalman filter to estimate the relative pose of two spacecraft via a nonlinear least-squares formulation. In a related vein, \cite{SCC.Volpe.Circi.ea2022} employs an unscented Kalman filter to infer relative pose and reconstruct the shape of the target using data from a monocular camera together with a range sensor. More recently, reinforcement-learning controllers have been explored for satellite inspection \cite{SCC.Wijk.Dunlap.ea2023, SCC.McQuinn.Dunlap.ea2024, SCC.Dunlap.Hamilton.ea2024}, where proximal policy optimization (PPO) is adopted as the training method. The PPO-trained policies control a single free-flying deputy carrying an optical sensor to carry out on-orbit inspection of a chief spacecraft. Because training occurs offline, the resulting controller is not readily adaptable in cluttered, debris-rich environments, which are precisely the scenarios that demand real-time adaptation.

Feature tracking during satellite inspection allows the deputy spacecraft to learn the structure of the chief satellite. While prior works like \cite{SCC.Huang.Jia.ea2019, SCC.Sun.Huo.ea2018, SCC.Li.Song.ea2022} have achieved navigation during satellite inspection, these approaches often rely on the assumption that the deputy spacecraft has full state feedback information for state feedback control. Image-based observers can be used to reconstruct the 3D environment \cite{SCC.Parikh.Cheng.ea2017, SCC.Bell.Chen.ea2017, SCC.Bell.Deptula.ea2020} and enable visual servo control \cite{SCC.Hutchinson.Hager.ea1996} using observed features, even without full-state feedback. A significant challenge in using camera-based feedback is the intermittent loss of object visibility caused by occlusions, feature tracking failures, or limited field-of-view (FOV).

Due to illumination and FOV constraints during satellite inspection, tracking features on the chief satellite may become impossible. To maximize the observability of features on the chief satellite, FOV constraints must be incorporated into the controller design. While image-based control techniques, such as those presented in \cite{SCC.Zhao.Emami.ea2021}, can be used to adjust camera focal length to maintain target visibility, their approach is not feasible for low-cost cameras with fixed focal lengths. Other approaches like \cite{SCC.Zhao.Zhang.ea2022} handle FOV constraints by developing an offline path planner using the gradient of a so-called navigation function. However, offline control approaches are not suitable for real-world navigation, as the controller cannot adapt in real time to changes in the operating environment.

Most of the techniques outlined above rely on the exact knowledge of the shape of the chief satellite to carry out inspection tasks. However, to achieve the 3D reconstruction of the chief satellite using a monocular camera, the methodology developed in this paper feeds pixel coordinates of features on the chief satellite into a memory regression (MRE)-based observer which estimates the Euclidean distances to these features to create a local map of the environment which is then used for relative navigation of the deputy spacecraft \cite{arXivSCC.Qureshi.Ogri.ea2025}. 
To carry out the inspection task, this map is fed to an adaptive optimal control-based planning and control algorithm that aims to maximize information gain from the camera. The maximization of the information gain, in this context, refers to improving the conditioning of the regressor matrix used by the MRE-based observer to compute distance estimates to features on the chief satellite. 

This paper builds on our earlier work \cite{SCC.Ogri.Qureshi.ea2025} and extends it to handle features entering and leaving the FOV of the camera, enabling accurate distance estimates and observer updates under changing visibility conditions. In this paper, we develop a switched MRE observer and dwell time conditions that ensure that even after the feature leaves the FOV, a new feature is used to localize the deputy spacecraft within the orbital environment. A k-means clustering algorithm is used to create goal states to aid the inspection of new areas and the controller regulates the deputy spacecraft to the $i$-th goal state. To prevent the deputy spacecraft from straying too far or colliding with the chief satellite, a robust barrier function is used to enforce safety, unlike the more conservative approach in \cite{SCC.Ogri.Qureshi.ea2025}, where a constant gain was used to approximate the Lagrange multiplier. In this paper, the robust barrier function explicitly incorporates a bound on the distance estimation error, which is guaranteed under the proposed MRE observer. This integration ensures provably safe relative motion of the deputy spacecraft while performing on-orbit inspection of the chief satellite.


\section{Problem Formulation}\label{section:problemFormulation}
\begin{figure}
    \centering
     \input{figures/cameraSchematic}
    \caption{Schematic of the relative position of the deputy spacecraft $\mathcal{O}_{B}$, chief satellite $\mathcal{O}_{H}$, and the goal location $\mathcal{O}_{G}$ is shown in Hill’s reference Frame.}
    \label{fig:cameraFeatureTracking}
\end{figure}
\subsection{Coordinate Frames of Spacecraft Motion Model}
The center of mass of the deputy spacecraft is assumed to coincide with the principal point of the current image captured by an onboard camera. $\mathcal{O}_{B}$ denotes the origin of the body frame, $\mathcal{F}_{B}$, with the basis denoted as $\left\{\vec{x}_{B}, \vec{y}_{B}, \vec{z}_{B}\right\}$, where $\vec{x}_{B}$ is
aligned with the horizontal of the image plane (i.e., to the right in the image), $\vec{y}_{B}$ is aligned with the vertical of the image plane (i.e., downward in the image), and $\vec{z}_{B}$ is aligned with the normal to the image plane. The body frame at the initial time $t_{0}$ is designated as the key frame $\mathcal{F}_{K}$, with origin $\mathcal{O}_{K}$, at the principal point of the image and basis $\left\{\vec{x}_{K}, \vec{y}_{K}, \vec{z}_{K}\right\}$, where $\vec{x}_{K}$ points to the right along the horizontal axis of the image plane, $\vec{y}_{K}$ points downward along the vertical axis of the image plane, and $\vec{z}_{K}$ points to the chief satellite, aligned with the optical axis of the camera.
The dynamics of the relative motion between the deputy spacecraft and the chief satellite are represented in Hill’s reference frame $\mathcal{F}_{H}$ as illustrated in Figure~\ref{fig:cameraFeatureTracking}. The basis $\left\{\vec{x}_{H}, \vec{y}_{H}, \vec{z}_{H}\right\}$ and the origin $\mathcal{O}_{H}$ centered on the center of mass chief satellite moves in a circular orbit around the earth. $\vec{x}_{H}$ denotes the unit vector aligned with the line joining the center of the Earth and $\mathcal{O}_{H}$, directed away from the Earth, $\vec{y}_{H}$ points in the tangential direction of the motion of the chief satellite, and $\vec{z}_{H}$-axis completes the right-handed basis. Finally, let $\mathcal{F}_{G}$ represent a goal frame with its origin at $\mathcal{O}_{G}$, and the basis $\left\{\vec{x}_{G}, \vec{y}_{G}, \vec{z}_{G}\right\}$.

\subsection{Relative dynamics between the deputy spacecraft and chief satellite}

To define the relative dynamics between the deputy spacecraft and the chief satellite, it is assumed that the chief satellite is in a near-circular orbit around the Earth so that the angular velocity of the Hill's reference frame is equivalent to the orbital rate of the chief satellite. Additionally, it is assumed that the chief satellite and the deputy spacecraft start close to each other, in near-synchronous orbits, so that only small maneuvers are needed to achieve relative navigation. Under these conditions and assuming a near-spherical Earth, the relative motion between the deputy spacecraft and chief satellite can be computed using the Clohessy-Wiltshire (CW) equations \cite{SCC.Clohessy.Wiltshire.ea1960},
\begin{align}\label{eq:deputyDynamics1}
\left(\dot{\vec{v}}_{b/h}\right)_{x} &=  2n \left(\vec{v}_{b/h}\right)_{y} + 3n^{2}\left(\vec{p}_{b/h}\right)_{x} + \frac{\left(F\right)_{x}}{m},\\
\label{eq:deputyDynamics2}
\left(\dot{\vec{v}}_{b/h}\right)_{y} &= - 2n\left(\vec{v}_{b/h}\right)_{x} + \frac{\left(F\right)_{y}}{m},\\
\left(\dot{\vec{v}}_{b/h}\right)_{z}  &= - n^{2}\left(\vec{p}_{b/h}\right)_{z} + \frac{\left(F\right)_{z}}{m},\label{eq:deputyDynamics3}
\end{align}
where $\vec{p}_{b/h}$ and $\vec{v}_{b/h}$ represent the relative position and velocity of the deputy spacecraft relative to the chief satellite, respectively, coordinatized in $\mathcal{F}_{H}$, $n = \frac{GM_{\oplus}}{r_{0}^3}$ is the mean motion of the chief satellite, $GM_{\oplus}$ is the gravitational parameter of the Earth, $F \in \mathbb{R}^{3}$ is the net force applied by the thrusters on the deputy spacecraft at the center of mass, $r_{0}$ is the radius of the orbit of the chief satellite, and $m$ is the mass of the deputy spacecraft.

\subsection{Sun Dynamics}

Sun is the only light source assumed in this paper, confined to the $\vec{x}_{H} - \vec{y}_{H}$ plane in Hill's reference frame. The unit vector pointing from the center of the
chief satellite to the Sun is denoted as $\vec{u}_{s/h}(t)$, and computed as $\vec{u}_{s/h}(t) = \begin{bmatrix}
    \cos(\theta_{s}(t)) & \sin(\theta_{s}(t)) & 0
\end{bmatrix}^{\top}$ at any time $t$. Here $\theta_{s}(t)$ is the angle formed by the conical ray of the sun on the surface of the chief satellite as depicted in Figure~\ref{fig:sun}. The angle $\theta_{s}$ is dependent on the mean
motion of the chief satellite $n$ such that \cite{SCC.Wijk.Dunlap.ea2023, SCC.McQuinn.Dunlap.ea2024, SCC.Dunlap.Hamilton.ea2024}
\begin{equation}\label{eq:sunDynamics}
    \dot{\theta}_{s} = -n, \quad \theta_{s}(0) = \theta_{s_{0}}.
\end{equation} 
\begin{assumption}\label{ass:unObstructedRays}
 There are no obstructions blocking the light emanating from the sun from reaching the chief satellite at any time $t$. 
\end{assumption}
\begin{figure}
    \centering
    \input{figures/sunIncidence}
    \caption{Visualization of the light emanating from the sun on the points on the surface of the chief satellite: Yellow dots represent inspection points that are illuminated by the sun, while blue dots indicate points that are not illuminated.}
    \label{fig:sun}
\end{figure}

\subsection{Motion model of the camera using inspection points around the satellite}	
    The following assumptions are required to facilitate the development of a model of the camera attached to the deputy spacecraft.
\begin{assumption}\label{ass:trackableFeatures}
    The chief satellite has a set of $N > 0$ features, $\mathcal{N} = \{1,\hdots,N\}$, that can be identified, provided they are illuminated by the Sun and are within the FOV of the camera. Moreover, $\forall t \in \mathbb{R}_{\geq 0}$, there exists a set of at least four trackable planar illuminated features and at least five non-planar illuminated features on the chief satellite that are within the FOV of the camera. Furthermore, the camera attached to the deputy spacecraft always faces the chief satellite. 
\end{assumption}
\begin{remark}
      Feature extraction techniques like those described in \cite{SCC.Bouguet.ea2001, SCC.Lucas.Kanade1981} can be used to extract features from images of the chief satellite and these features can be tracked provided they are within the FOV of the camera. We require at least four trackable planar illuminated features and at least five non-planar illuminated features to compute the normal vector pointing outward from the centroid of the image-plane locations of the features currently within the FOV of the camera, denoted by $\vec{n}_{c}\in \mathbb{R}^{3}$.
      
      To ensure that the deputy spacecraft always faces the satellite, one can leverage onboard attitude control systems, such as reaction wheels, IMUs, and attitude control laws \cite{SCC.Chung.Bandyopadhyay.ea2013,SCC.Nakka.Hoenig.ea2022}.
\end{remark}
\begin{assumption}\label{ass:intrinsicMatrix}
    The intrinsic matrix of the camera, $\mathcal{A} \in \mathbb{R}^{3 \times 3}$, is known and invertible \cite{SCC.Ma.Soatto.ea2004}.
\end{assumption}
  Under Assumptions~\ref{ass:trackableFeatures} and \ref{ass:intrinsicMatrix}, the unit vector pointing from origin of the deputy spacecraft to the $i$-th feature of the chief, denoted by $\vec{u}_{b/h_{i}}(t) \in \mathbb{R}^{3}$, and the fixed unit vector pointing from the origin of the key frame to the $i$-th feature on the chief satellite, denoted by $\vec{u}_{k/h_{i}} \in \mathbb{R}^{3}$, can be obtained as $\vec{u}_{b/h_{i}}(t) =  \frac{\mathcal{A}^{-1}P_{b/h_{i}}(t)}{\big\|\mathcal{A}^{-1}P_{b/h_{i}}(t)\big\|}$ and $\vec{u}_{k/h_{i}} =  \frac{\mathcal{A}^{-1}P_{k/h_{i}}}{\big\|\mathcal{A}^{-1}P_{k/h_{i}}\big\|}$, respectively, where $P_{b/h_{i}}(t)$ and $P_{k/h_{i}}$ are the measured pixel coordinates of the $i$-th feature on the chief satellite relative to the deputy spacecraft and key frame in $\mathcal{F}_{H}$, respectively, represented as homogeneous coordinates. 
\begin{assumption}\label{ass:goalMeasurements}
    The position of the goal relative to the chief satellite is known to the deputy spacecraft.
\end{assumption}
\begin{remark}
    The unit vector pointing from the nearest cluster of uninspected points on the chief satellite to the goal, denoted by $\vec{u}_{g/h} \in \mathbb{R}^{3}$, is determined through k-means clustering \cite{SCC.MacQueen.ea1967}, which serves as a high-level planner. This unit vector is relayed intermittently by the k-means algorithm to the deputy spacecraft. Using the unit vector from the k-means algorithm, the position of the goal relative to the chief satellite, coordinatized in $\mathcal{F}_{H}$, is computed by the deputy spacecraft as $\vec{p}_{g/h} = \vec{u}_{g/h}r_{g/h} \in \mathbb{R}^{3}$, where $r_{g/h} > 0$ is a user-defined desired distance between the goal and the chief satellite.
\end{remark}

Under Assumptions~\ref{ass:trackableFeatures}--\ref{ass:goalMeasurements}, the positions of the deputy spacecraft, goal, key frame, and the $i$-th feature on the chief satellite relative to each other can be expressed as (see Figure~\ref{fig:cameraFeatureTracking})
\begin{equation}\label{eq:goalPosition}
    \vec{p}_{g/b}(t) =  \vec{p}_{g/h_{i}} - \vec{p}_{b/h_{i}}(t),
\end{equation}
and
\begin{equation}\label{eq:deputyKeyFrame}
    \vec{p}_{b/h_{i}}(t) - \vec{p}_{b/k}(t) = \vec{p}_{k/h_{i}},
\end{equation}
where $\vec{p}_{a/b}$ denotes position of $a$ relative to $b$, coordinatized in $\mathcal{F}_{H}$, for any $a$ and $b$. Note, all position and velocity vectors in this paper are coordinatized in $\mathcal{F}_{H}$. The relative position vectors between the deputy spacecraft, the key frame, the goal, and the $i$-th feature of the chief satellite can be expressed in terms of unit vector directions as
\begin{align}
    \vec{p}_{b/h_{i}}(t) &= r_{b/h_{i}}(t) \vec{u}_{b/h_{i}}(t),\\
    \vec{p}_{g/b}(t) &= r_{g/b}(t) \vec{u}_{g/b}(t),\\
    \vec{p}_{b/k}(t) &= r_{b/k}(t) \vec{u}_{b/k}(t),\\
    \vec{p}_{k/h_{i}} &= r_{k/h_{i}} \vec{u}_{k/h_{i}},
\end{align}
where the scalars $r_{b/h_{i}}(t) > 0$, $r_{g/b}(t) > 0$, $r_{b/k}(t) > 0$, and $r_{k/h_{i}} > 0$ are unknown Euclidean distances between the deputy spacecraft, key frame, chief satellite and the $i$-th point on the chief satellite. The unit vectors $\vec{u}_{b/h_{i}}(t), \vec{u}_{b/k}(t), \text{and } \vec{u}_{k/h_{i}} \in \mathbb{R}^{3}$ are known under Assumption~\ref{ass:intrinsicMatrix}, but the unit vector pointing from the goal to the deputy spacecraft $\vec{u}_{g/b}(t) \in \mathbb{R}^{3}$ is assumed unknown, which implies that $\vec{p}_{g/b}(t)$ is unknown.

\subsection{Localization Objective and Sequential Feature Tracking for Relative Navigation}\label{section:sequentialTracking}

The main objectives of the deputy spacecraft are to reach the goal and to localize itself relative to the goal, i.e., to generate estimates $\hat{\vec{p}}_{g/b}$ of $\vec{p}_{g/b}$, while maximizing information gain from the camera. 
In this paper, a robust observer is designed to estimate the needed distances. Due to transient behavior of the observer, convergence of $\hat{\vec{p}}_{g/b}$ to $\vec{p}_{g/b}$ requires that, on average, features used for localization remain in the FOV of the camera for a sufficient amount of time (referred to hereafter as the dwell-time).

If the chief satellite is large in size and motion of the deputy spacecraft relative to the chief satellite is fast, individual features may not stay within the FOV of the camera long enough for convergence of $\hat{\vec{p}}_{g/b}$ to $\vec{p}_{g/b}$  within a desired error bound $\epsilon > 0$\footnote{The estimated relative position vector $\hat{\vec{p}}_{a/b}(t) \in \mathbb{R}^{3}$ can be computed as $\hat{\vec{p}}_{a/b}(t) =  \hat{r}_{a/b}(t) \vec{u}_{a/b}(t)$, where $r_{a/b}(t) \in \mathbb{R}_{\geq 0}$ is the estimated distance between $a$ and $b$, and $\vec{u}_{a/b}(t) \in \mathbb{R}^{3}$ is the unit vector pointing from $a$ to $b$.}. To account for features leaving the FOV of the camera, we design a sequential tracking approach based on daisy-chaining distance estimates obtained from multiple features as they enter and leave the FOV.


To continue improving the estimate $\hat{\vec{p}}_{g/b}(t)$ and ultimately achieve the desired $\epsilon$-accuracy, the camera begins tracking a new feature as it enters its FOV and switches to the new feature once the dwell time is met and the current feature has left the FOV. By chaining together local position estimates from overlapping tracking intervals, the observer achieves convergence of $\hat{\vec{p}}_{g/b}(t)$ within the desired error tolerance.

Figure~\ref{fig:dwellTime} illustrates the sequential feature tracking process. Let $t_{i}^{a} \in \mathbb{R}_{\geq 0}$ denote the time at which tracking begins for feature $i$ on the chief satellite, $t_{i}^{d} \in \mathbb{R}_{\geq 0}$ the time when the dwell time condition for feature $i$ is satisfied, and $t_{i}^{u} \in \mathbb{R}_{\geq 0}$ the time at which tracking of feature $i$ ends. Letting $t_{1}^{a} = 0$, the relative position $\vec{p}_{k/h_{1}}$ is estimated over the interval $[0, t_{1}^{u}]$, and used to compute $\hat{\vec{p}}_{g/b}(t)$. For the subsequent interval $[t_{1}^{u}, t_{1}^{d}]$, the estimate $\hat{\vec{p}}_{k/h_{1}}(t_{1}^{u})$ is held constant and continues to be used to determine $\hat{\vec{p}}_{g/b}(t)$.

Subsequently, as depicted in Figure~\ref{fig:dwellTime}, when a second feature becomes visible at time $t_{2}^{a} \in \mathbb{R}_{\geq 0}$ (which may occur before $t_{1}^{u}$), the relative position $\vec{p}_{k/h_{2}}(t)$ is estimated over the interval $[ t_{2}^{a}, t_{2}^{d}]$, where $t_{2}^{d}$ is selected to achieve an error smaller than that from feature $1$, i.e., $\|\tilde{\vec{p}}_{k/h_{2}}(t_{2}^{d})\| \leq \|\tilde{\vec{p}}_{k/h_{1}}(t_{1}^{u})\|$. Since the error $\tilde{\vec{p}}_{k/h_{1}}(t)$ is not measurable, we employ a Lyapunov-based dwell time condition that relies on known bounds on the error, estimated using decay rates and initial conditions, to verify whether the previous inequality is satisfied. Once the dwell time is satisfied, estimates of $\vec{p}_{k/h_{2}}$ are used to further refine the estimate $\hat{\vec{p}}_{g/b}(t)$ over the interval $[ t_{2}^{d}, t_{2}^{u}]$. This process repeats across successive features until the position estimation error $\tilde{\vec{p}}_{g/b}(t)$ is within the specified $\epsilon$-error bound.
\begin{figure}
\centering
\input{figures/dwellTime}
\caption{Sequential feature-based localization with feature switching upon FOV exit after satisfying dwell time.}
\label{fig:dwellTime}
\end{figure}

To achieve the aforementioned objectives, this paper develops an MRE-based distance observer to estimate the goal position relative to the camera, $\vec{p}_{g/b}(t)$, while features on the chief satellite enter and leave the FOV of the camera attached to the deputy. The dwell time condition is developed under which the observer switches to new features to improve estimates of $\vec{p}_{g/b}(t)$. The estimates are then used to simultaneously synthesize and utilize a controller that maximizes information gain, as quantified in the constrained optimal problem in \eqref{eq:coop}, while ensuring local uniform boundedness of the trajectories of the closed-loop system.
\begin{remark}
If the time instance $t_{i}^{d}$ for a feature $i$ exceeds the time $t_{i}^{u}$ at which the feature $i$ exits the field of view, we disregard this feature for localization purposes. 
\end{remark}

\section{Observer Update Laws For Estimating Euclidean Distances}\label{section:MREObserver}
In this section, an MRE-based distance observer is developed to estimate the unknown Euclidean distances $r_{b/h_{i}}(t)$, $r_{b/k}(t)$, and $r_{k/h_{i}}$. The expression in \eqref{eq:deputyKeyFrame} be equivalently written as
\begin{equation}\label{eq:initialDistanceRelationShip}
    Y_{i}(t)\begin{bmatrix} r_{b/h_{i}}(t) \\ r_{b/k}(t)\end{bmatrix} = \vec{u}_{k/h_{i}}r_{k/h_{i}}.
\end{equation}
where $Y_{i}(t) = \begin{bmatrix}
\vec{u}_{b/h_{i}}(t) & -\vec{u}_{b/k}(t)
\end{bmatrix} \in \mathbb{R}^{3 \times 2}$. Since illuminated features on the chief satellite may leave the FOV of the camera as the deputy spacecraft moves towards the goal location, the matrix $Y_{i}(t)$ in \eqref{eq:initialDistanceRelationShip} may become unmeasurable. 
 To account for illuminated features of the chief satellite entering or leaving the FOV of the camera, let $\mathcal{S}(t) \subseteq \mathcal{N}$ represent the set
of illuminated features within the FOV of the camera at time $t$, given by 
\begin{equation}
\scalemath{0.98}{\mathcal{S}(t) = \left\{i \in \mathcal{N} \,\middle|\, 
\begin{aligned}
&r_{\min} \leq \|r_{b/h_{i}}(t)\| \leq r_{\max}, \, \\
&\cos^{-1}\left( \frac{\langle \vec{p}_{b/h_{i}}(t), \vec{n}_{c} \rangle}{\|\vec{p}_{b/h_{i}}(t)\| \|\vec{n}_{c}\|} \right) \leq \alpha_{\text{FOV}}
\end{aligned}
\right\}},
\end{equation}
where the cone angle $\alpha_{\text{FOV}}$ is the conical FOV of the camera as described in Figure~\ref{fig:sun}, $r_{\min} \in \mathbb{R}_{>0}$ denotes the minimum allowable safe distance between the deputy spacecraft and the chief satellite, $r_{\max} \in \mathbb{R}_{> 0}$ is maximum sensing range of the camera, $\vec{n}_{c} \in \mathbb{R}^{3}$ is a unit vector normal to the surface of the chief satellite emanating from the centroid of the cluster of features that are currently within the FOV of the camera and the notation $\langle \cdot, \cdot \rangle$ denotes the dot product. Let the minimum allowable safe distance between the deputy spacecraft and the chief satellite be defined as $r_{\min} \coloneqq r_{d} + r_{c}$, where $r_{d}$ and $r_{c} > 0$ are the radii of the deputy spacecraft and chief satellite, respectively. Under the controller designed in Section~\ref{section:controlDesign}, the deputy spacecraft is guaranteed to never collide with or stray too far from the chief satellite. Hence, the distance between the deputy spacecraft and the $i$-th feature on the chief satellite satisfies $r_{\min} \leq r_{b/h_{i}}(t) \leq r_{\max}$ for all $i \in \mathcal{N}$.
The complement of $\mathcal{S}(t)$, given by $\mathcal{S}^{c}(t) = \mathcal{N}\setminus\mathcal{S}(t)$, represents the set of illuminated features outside the FOV of the camera. Let the time interval $\Delta t_{i}^{a} = t_{i}^{u} - t_{i}^{a} > 0$ represent the time duration the $i$-th feature on the chief satellite spends in $\mathcal{S}(t)$.
\begin{assumption}\label{ass:startingFeature}
    The first illuminated feature is in the FOV of the camera upon initialization, i.e. $1 \in \mathcal{S}(t)$ when $ t_{1}^{a} = 0$ (cf. \cite{SCC.Parikh.Kamalapurkar.ea2018, SCC.Bell.Sun.ea2022}).
\end{assumption}
\begin{remark}
    Since the deputy spacecraft is always facing the chief satellite under Assumption~\ref{ass:trackableFeatures}, Assumption~\ref{ass:startingFeature} can be satisfied by ensuring the deputy spacecraft is facing the illuminated part of the chief satellite upon initialization.
\end{remark}
\begin{assumption}\label{ass:originsNotCoincident}
    The origins $\mathcal{O}_{B}$ and $\mathcal{O}_{K}$ are not coincident while $t > t_{i}^{a}$, specifically, there exists a constant $\underline{r} > 0$ such that $r_{b/k}(t) > \underline{r}$ for all $t \in (t_{i}^{a},  \; t_{i}^{u})$.
\end{assumption}
\begin{remark}
    To satisfy Assumption \eqref{ass:originsNotCoincident}, key frames are selected during implementation to ensure that the camera origin $\mathcal{O}_{B}$ and the key frame origin $\mathcal{O}_{K}$ are sufficiently far apart, thus preventing them from being coincident at any time t.
\end{remark}
\begin{assumption}\label{ass:knownLinearVelocity}
     The linear velocity of the camera relative to the chief satellite, $\vec{v}_{b/h}(t)$, coordinatized in $\mathcal{F}_{H}$, is measurable and bounded by a constant $v_{\max} > 0$ such that it satisfies $\|\vec{v}_{b/h}(t)\| \leq v_{\max}$.
\end{assumption}
If $Y_{i}(t)$ is full column rank, the expression in \eqref{eq:initialDistanceRelationShip} can be simplified to obtain
\begin{equation}\label{eq:distanceRegressor}
    \begin{bmatrix}
        r_{b/h_{i}}(t) \\ r_{b/k}(t)
    \end{bmatrix} = \psi_{i}(t) r_{k/h_{i}},
\end{equation}
where $\psi_{i}(t) = \left(Y_{i}^{\top}(t)Y_{i}(t)\right)^{-1}Y_{i}^{\top}(t)\vec{u}_{k/h_{i}} \in \mathbb{R}^{2 \times 1}$.
Under Assumptions~\ref{ass:intrinsicMatrix}--\ref{ass:knownLinearVelocity}, the time derivatives of the unknown distances are given, for all $i \in \mathcal{S}(t)$ and $t \in \mathbb{R}_{\geq 0}$, by 
\begin{equation}\label{eq:unknownDistDynamics1}
    \dot{r}_{b/h_{i}}(t) = \vec{u}_{b/h_{i}}^{\top}(t)\vec{v}_{b/h}(t),
\end{equation}
\begin{equation}\label{eq:unknownDistDynamics2}
    \dot{r}_{b/k}(t) = \vec{u}_{b/k}^{\top}(t)\vec{v}_{b/h}(t),
\end{equation}
and
\begin{equation}\label{eq:unknownDistDynamics3}
\dot{r}_{k/h_{i}} = 0.
\end{equation}

To estimate the unknown distances $r_{b/h_{i}}$, $r_{b/k}$, and $r_{k/h_{i}}$, this paper employs memory-based adaptive control.

\subsection{Learning the Feature Structure of Chief Satellite Using Memory-Based Regression}
To facilitate the development of the MRE-based distance observer, let $\Xi_{i} \coloneqq \begin{bmatrix}
\vec{u}_{b/h_{i}}^{\top}(t) & \vec{u}_{b/k}^{\top}(t)
\end{bmatrix}^{\top}\vec{v}_{b/h}(t)$. The first step in the design of an MRE-based observer is to express the dynamics in \eqref{eq:unknownDistDynamics1}--\eqref{eq:unknownDistDynamics3} as $\mathcal{Y}_{i}(t)r_{k/h_{i}} = \mathcal{U}_{i}(t), \, t > t_{i}^{a}$, where $\mathcal{Y}_{i}:\mathbb{R}_{\geq 0} \to \mathbb{R}^{2}$ and $\mathcal{U}_{i}:\mathbb{R}_{\geq 0} \to \mathbb{R}^{2}$ are signals that can be computed using the available measurements.

One technique for MRE-based observer design involves using
regressor filtering, where the dynamics in \eqref{eq:unknownDistDynamics1} and \eqref{eq:unknownDistDynamics3} are expressed in filtered form.
To retain memory of the regressor signal, this paper adopts a exponentially decaying filter $\phi_{i}(t) \coloneqq \lambda_{i}e^{-\lambda_{i}t}$, where the constant $\lambda_{i} \in \mathbb{R}_{> 0}$ is the forgetting factor.
Applying the time-varying filter $\phi_{i}(t)$ to both sides of \eqref{eq:unknownDistDynamics1} and \eqref{eq:unknownDistDynamics2}, we get the filtered form as
\begin{equation}\label{eq:distanceFilteredForm}
     \begin{bmatrix}
        r_{b/h_{i}}^{\mathtt{f}}(t) \\ r_{b/k}^{\mathtt{f}}(t)
    \end{bmatrix} = \Xi_{i}^{\mathtt{f}}(t), 
\end{equation}
where the signals $r_{b/h_{i}}^{\mathtt{f}}(t)
 \coloneqq \phi_{i}(t)*[\dot{r}_{b/h_{i}}(t)]$, $r_{b/k}^{\mathtt{f}}(t)
 \coloneqq \phi_{i}(t)*[\dot{r}_{b/k}(t)]$, and $\Xi_{i}^{\mathtt{f}}(t)
 = \phi_{i}(t)*[\Xi_{i}(t)]$ denote the filtered versions of their respective inputs, obtained by applying the filter $\phi_{i}(t)$. Let $r_{i}^{\mathtt{f}}(t) \coloneqq \begin{bmatrix}
    r_{b/h_{i}}^{\mathtt{f}}(t) & r_{b/k}^{\mathtt{f}}(t)
\end{bmatrix}^{\top} \in \mathbb{R}^{2}$ and $r_{i}(t) \coloneqq \begin{bmatrix}
    r_{b/h_{i}}(t) & r_{b/k}(t)
\end{bmatrix}^{\top} \in \mathbb{R}^{2}$. Since convolution with an exponential is equivalent to solving a first-order differential equation, we can write
{\medmuskip=2mu\thinmuskip=2mu\thickmuskip=2mu\begin{align}\label{eq:filteredForm2}
r_{i}^{\mathtt{f}}(t) &= \int_{t_{i}^{a}}^{t}\phi_{i}(t-\tau)\dot{r}_{i}(\tau)\, \mathrm{d}\tau 
= \lambda_{i} \int_{t_{i}^{a}}^t e^{-\lambda_{i} (t - \tau)} \dot{r}_{i}(\tau)\, \mathrm{d}\tau \nonumber \\
&= -r_{i}(t) + e^{-\lambda_{i} t} r_{i}(t_{i}^{a}) + \lambda_{i} \int_{t_{i}^{a}}^t e^{-\lambda_{i} (t - \tau)} r_{i}(\tau)\, \mathrm{d}\tau.
\end{align}}
By substituting \eqref{eq:distanceRegressor} and \eqref{eq:distanceFilteredForm} into \eqref{eq:filteredForm2}, we obtain the vectors $\mathcal{Y}_{i}(t) \in \mathbb{R}^{2}$ and $\mathcal{U}_{i}(t) \in \mathbb{R}^{2}$ as 
\begin{equation}
    \mathcal{Y}_{i}(t) \coloneqq \begin{cases}
    -\psi_{i}(t) + \psi_{i}^{e}(t, t_{i}^{a}) , & t \in (t_{i}^{a}, \; t_{i}^{u}], \\
    0_{2\times 1}, & t \in (t_{i}^{u}, \; t_{i+1}^{a}],
\end{cases}
\end{equation}
and 
\begin{equation}
    \mathcal{U}_{i}(t) \coloneqq \begin{cases}
    \Xi_{i}^{\mathtt{f}}(t), & t \in (t_{i}^{a}, \; t_{i}^{u}], \\
    0_{2\times 1}, & t \in (t_{i}^{u}, \; t_{i+1}^{a}],
\end{cases}
\end{equation}
respectively, where $\psi_{i}^{e}(t, t_{i}^{a}) = e^{-\lambda_{i} t} \psi_{i}(t_{i}^{a}) + \lambda_{i} \int_{t_{i}^{a}}^t e^{-\lambda_{i} (t - \tau)} \psi_{i}(\tau)\, \mathrm{d}\tau$.

Another MRE approach uses windowed integration, where
the dynamics in \eqref{eq:unknownDistDynamics1} and \eqref{eq:unknownDistDynamics2} are integrated over a finite time window of
length $\varsigma > 0$. Integrating \eqref{eq:unknownDistDynamics1} and \eqref{eq:unknownDistDynamics2} over $\varsigma$, we get
\begin{equation}\label{eq:secondDistance}
    \begin{bmatrix}
        r_{b/h_{i}}(t) \\ r_{b/k}(t) 
    \end{bmatrix} - \begin{bmatrix}
        r_{b/h_{i}}(t-\varsigma) \\ r_{b/k}(t-\varsigma) 
    \end{bmatrix} =  \int_{t-\varsigma}^{t} \Xi_{i}(\tau)\,  \mathrm{d}\tau, 
\end{equation}
with $t > t_{i}^{a} + \varsigma$. Substituting equation \eqref{eq:secondDistance} into \eqref{eq:distanceRegressor} at times $t$ and $t-\varsigma$ also yields the vectors $\mathcal{Y}_{i}(t) \in \mathbb{R}^{2}$ and $\mathcal{U}_{i}(t) \in \mathbb{R}^{2}$ as \begin{equation}\mathcal{Y}_{i}(t) = \begin{cases}
    \left(\psi_{i}(t)-\psi_{i}(t_{i}^{a})\right), & t \in \left(t_{i}^{a}, \; t_{i}^{a} + \varsigma\right], \\
    \left(\psi_{i}(t)-\psi_{i}(t-\varsigma)\right), & t \in \left(t_{i}^{a} + \varsigma, \; t_{i}^{u}\right],\\
    0_{2 \times 1}, & t \in \left(t_{i}^{u}, \; t_{i+1}^{a}\right],
\end{cases}
\end{equation}
and \begin{equation}\mathcal{U}_{i}(t) = \begin{cases}
    \int_{t_{i}^{a}}^{t} \Xi_{i}(\tau)\,  \mathrm{d}\tau,  & t \in \left(t_{i}^{a}, \; t_{i}^{a} + \varsigma\right], \\
    \int_{t-\varsigma}^{t} \Xi_{i}(\tau)\,  \mathrm{d}\tau,  & t \in \left(t_{i}^{a} + \varsigma, \; t_{i}^{u}\right], \\
    0_{2 \times 1}, & t \in \left(t_{i}^{u}, \; t_{i+1}^{a}\right],
\end{cases}
\end{equation} where $t \in \left[t_{i}^{a}, t_{i}^{u}\right)$ implies that the $i$-th feature on the chief satellite is in the set $\mathcal{S}(t)$ and $t \in \left[t_{i}^{u}, t_{i+1}^{a}\right)$ implies that the $i$-th feature on the chief satellite is in the set $\mathcal{S}^{c}(t)$. 

Multiplying both sides of the relation above by $\mathcal{Y}_{i}^{\top}(t)$ yields
 \begin{equation}\label{eq:augDynamics2}
    \mathcal{Y}_{i}^{\top}(t)\mathcal{Y}_{i}(t)r_{k/h_{i}} = \mathcal{Y}_{i}^{\top}(t)\mathcal{U}_{i}(t)
\end{equation}
In general, $\mathcal{Y}_{i}(t)$ will not have full column rank while $\sigma_{i}(t) = a$ (e.g. when the camera
is stationary) implying that in general $\mathcal{Y}_{i}^{\top}(t)\mathcal{Y}_{i}(t) \succeq 0$. However, {multiple measurements  can be summed together as
\begin{equation}\label{eq:CLdynamics}
    \sum_{s = 1}^{M} \frac{\mathcal{Y}_{i}^{\top}(t_{s})\mathcal{Y}_{i}(t_{s})}{1+\|\mathcal{Y}_{i}(t_{s})\|^{2}}
    r_{k/h_{i}} = \sum_{s = 1}^{M} \frac{\mathcal{Y}_{i}^{\top}(t_{s})\mathcal{U}_{i}(t_{s})}{1+\|\mathcal{Y}_{i}(t_{s})\|^{2}},
\end{equation}
where $M \in \mathbb{Z}_{\geq 1}$ and $t_{s} \in \left(t_{i}^{a}, \; t_{i}^{u}\right]$. Let $\Sigma_{\mathcal{Y}_{i}} \coloneqq \sum_{s = 1}^{M} \frac{\mathcal{Y}_{i}^{\top}(t_{s})\mathcal{Y}_{i}(t_{s})}{{1+\|\mathcal{Y}_{i}(t_{s})\|^{2}}}$ and $\Sigma_{\mathcal{U}_{i}} \coloneqq \sum_{s = 1}^{M} \frac{\mathcal{Y}_{i}^{\top}(t_{s})\mathcal{U}_{i}(t_{s})}{1+\|\mathcal{Y}_{i}(t_{s})\|^{2}}$. To facilitate the subsequent stability analysis in Section~\ref{section:observerStabilityAnalysis}, the following rank condition on the regressor matrix $\Sigma_{\mathcal{Y}_{i}}$ is necessary.
\begin{assumption}\label{ass:sufficientExcitation}
    The motion of the camera is assumed to be uniformly sufficiently rich such that there exists a set of features $\mathcal{V}(t) \subset \mathcal{N}$ satisfying $\mathcal{V}(t)\cap\mathcal{S}(t)\neq \emptyset$, a constant $\underline{\Sigma} > 0$, and a set of finite time instances $\mathcal{T} = \left\{t_{i}^{*}\right\}_{i \in \mathcal{V}(t)}$ such that for all $t \geq t_{i}^{*}$ with $t_{i}^{*} \in \left(t_{i}^{a}, \; t_{i}^{u}\right)$, and all initial conditions $r_{b/h_{i}}(0), r_{b/k}(0)$, $\vec{u}_{b/h_{i}}(0)$, and $\vec{u}_{b/k}(0)$,
    \begin{equation}\label{eq:rankCond}
\lambda_{\min}\left(\Sigma_{\mathcal{Y}_{i}}\right) > \underline{\Sigma}, \quad \forall i \in \mathcal{V}(t).
    \end{equation}
\end{assumption}
\begin{remark}
    Assumption \ref{ass:sufficientExcitation} is an observability condition required for the subsequent development that is similar to other image-based observers \cite{SCC.Bell.Chen.ea2017, SCC.Bell.Deptula.ea2020, SCC.Ogri.Qureshi.ea2024}. Assumption~\ref{ass:sufficientExcitation} is impossible to guarantee \emph{a piori}, but \eqref{eq:rankCond} can be verified online. 
\end{remark}

The time instances $t_{i}^{*}$ are unknown \emph{a piori}; however, they can be determined online by checking the minimum eigenvalue of the regressor matrix $\Sigma_{\mathcal{Y}_{i}}$ in Assumption~\ref{ass:sufficientExcitation}. After $t\geq t_{i}^{*}$, $\lambda_{\min}\{\Sigma_{\mathcal{Y}_{i}}
\} > \underline{\Sigma}$ implies that $r_{k/h_{i}} = \Sigma_{\mathcal{Y}_{i}}^{-1}\Sigma_{\mathcal{U}_{i}}$ for feature  $i \in \mathcal{V}(t)$. Let $\mathcal{P}(t) = \left\{i \in \mathcal{V}(t) \cap \mathcal{S}(t)\mid \lambda_{\min}\left(Y_{i}^{\top}(t)Y_{i}(t)\right) > \lambda_{a}\right\}$ and $\mathcal{P}^{c}(t) = \mathcal{N} \setminus \mathcal{P}(t)$. Let $\left\{\sigma_{i}(t)\right\}_{i \in \mathcal{S}}$ be the set of
switching signals, where $\sigma_{i}(t) \in \left\{u, a\right\}$ indicates whether $i \in \mathcal{P}(t)$ or $i \in \mathcal{P}^{c}(t)$. Specifically, $\sigma_{i}(t) = a$ implies that illuminated feature $i$ is in the FOV, the camera has sufficient motion such that the rank condition in Assumption~\ref{ass:sufficientExcitation} is satisfied $\lambda_{\min}\left(Y_{i}^{\top}(t)Y_{i}(t)\right) > \lambda_{a}$; while $\sigma_{i}(t) = u$ implies that the feature is outside the FOV of the camera or the rank condition in Assumption~\ref{ass:sufficientExcitation} is not satisfied or $\lambda_{\min}\left(Y_{i}^{\top}(t)Y_{i}(t)\right) \leq 0$.

Let $\theta_{i}(t) \coloneqq \begin{bmatrix}
    r_{b/h_{i}}(t)  & r_{b/k}(t)  & r_{k/h_{i}}
\end{bmatrix}^{\top} \in \mathbb{R}^{3}$ be a concatenated vector of the unknown distances. Substituting $r_{k/h_{i}} = \Sigma_{\mathcal{Y}_{i}}^{-1}\Sigma_{\mathcal{U}_{i}}$ into \eqref{eq:initialDistanceRelationShip} yields the relationship
\begin{equation}
    \mathscr{Y}_{i}(t)\theta_{i}(t) = \mathscr{U}_{i}(t),
\end{equation}
for all $t$ such that $\sigma_{i}(t) = a$, where $\mathscr{Y}_{i}(t) \coloneqq \begin{bmatrix}
    Y_{i}(t) & 0_{3 \times 1} \\ 0_{1 \times 2} & 1
\end{bmatrix} \in \mathbb{R}^{4 \times 3}$ and $\mathscr{U}_{i}(t) \coloneqq  \begin{bmatrix}
    \vec{u}_{k/h_{i}}\Sigma_{\mathcal{Y}_{i}}^{-1}\mathcal{U}_{i}(t) \\ \Sigma_{\mathcal{Y}_{i}}^{-1}\mathcal{U}_{i}(t)
\end{bmatrix} \in \mathbb{R}^{4}$.
Let $\hat{\theta}_{i}(t) = \begin{bmatrix}
    \hat{r}_{b/h_{i}}(t)  & \hat{r}_{b/k}(t)  & \hat{r}_{k/h_{i}}(t)
\end{bmatrix}^{\top} \in \mathbb{R}^{3}$ be the concatenated vector of the distance estimates. Based on the subsequent stability analysis in Section~\ref{section:observerStabilityAnalysis}, an MRE-based adaptive update law can be used to estimate the unknown distances estimates $\hat{\theta}_{i}(t) \in \mathbb{R}^{3}$ as
\begin{equation}\label{eq:updateLaw}
    \dot{\hat{\theta}}_{i}(t) = \begin{cases}
         0, & \sigma_{i}(t) = u,\\
         \proj\left(\Pi_{i}(\hat{\theta}_{i}(t), t)\right), & \sigma_{i}(t) = a,
    \end{cases}
\end{equation}
where $\Pi_{i}(\hat{\theta}_{i}(t), t) \coloneqq \mu_{i}(t) + k_{\theta}\mathscr{Y}_{i}^{\top}(t)\left(\mathscr{U}_{i}(t)-\mathscr{Y}_{i}(t)\hat{\theta}_{i}(t)\right)$,  $\mu_{i}(t) \coloneqq \begin{bmatrix}
        \vec{u}_{b/h_{i}}^{\top}(t)\vec{v}_{b/h}(t) & \vec{u}_{b/k}^{\top}(t)\vec{v}_{b/h}(t) & 0
        \end{bmatrix}^{\top}$, and $k_{\theta} \in \mathbb{R}_{>0}$ is a constant observer gain. The notation $\proj(\cdot)$ represents the smooth projection operator defined in \cite{SCC.Cai.Queiroz.ea2006, SCC.Krstic.Kanellakopoulos.ea1995} employed to ensure that the bounds $r_{\min} \leq \hat{r}_{b/h_{i}}(t) \leq r_{\max}$, $\underline{r} \leq \hat{r}_{b/k}(t)$, and $r_{\min} \leq \hat{r}_{k/h_{i}}(t) \leq r_{\max}$ hold for all $t \in \mathbb{R}_{> t_{i}^{a}}$. The MRE-based update laws are designed based on whether sufficient data have been collected based on the motion of the camera and whether the $i$-th feature is in the FOV. Let $\tilde{\theta}_{i}(t) \in \mathbb{R}^{3}$ represent the concatenated distance estimation error defined as 
\begin{equation}\label{eq:distanceError}
    \tilde{\theta}_{i}(t) \coloneqq \theta_{i}(t) - \hat{\theta}_{i}(t)
\end{equation}
Taking the time derivative of \eqref{eq:distanceError}, substituting the derivatives in \eqref{eq:unknownDistDynamics1}--\eqref{eq:unknownDistDynamics3}, and the adaptive update law in \eqref{eq:updateLaw} yields 
\begin{equation}\label{eq:observerErrorDynamics}
    \dot{\tilde{\theta}}_{i}(t) = \begin{cases}
         0, & \sigma_{i}(t) = u,\\
        -k_{\theta}\mathscr{Y}_{i}^{\top}(t)\mathscr{Y}_{i}(t)\tilde{\theta}_{i}, & \sigma_{i}(t) = a.
    \end{cases}
\end{equation}

 The following section analyses the convergence properties of the observer error system.

\section{Observer Design Stability Analysis}\label{section:observerStabilityAnalysis}

 To facilitate the analysis, let $\tau_{i}^{a} \in \mathbb{R}_{\geq 0}$ represent the time instance when $\sigma_{i}(t) = a$
and let $\tau_{i}^{u} \in \mathbb{R}_{\geq 0}$ represent the time instance when $\sigma_{i}(t) = u$. Let $\Delta \tau_{i}^{a} \coloneqq \tau_{i}^{u} - \tau_{i}^{a}$ and $\Delta \tau_{i}^{u} \coloneqq \tau_{i+1}^{a} - \tau_{i}^{u}$ be the duration where $\sigma_{i}(t) = a$ and $\sigma_{i}(t) = u$, respectively. Let $W_{i}: \mathbb{R}^{3} \to \mathbb{R}$ be a candidate Lyapunov function defined as 
\begin{equation}\label{eq:observerLyap}
    W_{i}(\tilde{\theta}_{i}) \coloneqq \frac{1}{2}\tilde{\theta}_{i}^{\top}\tilde{\theta}_{i}.
\end{equation}
 The following theorem establishes global uniform ultimate boundedness of the observer error system.
\begin{theorem}\label{thm:expObserver}
     If Assumptions~\ref{ass:unObstructedRays}-\ref{ass:sufficientExcitation} hold and the trajectories of \eqref{eq:unknownDistDynamics1}--\eqref{eq:unknownDistDynamics3} does not blow up to infinity in finite time, the MRE-based update law defined in \eqref{eq:updateLaw} ensures that the trajectories of the concatenated observer error $\tilde{\theta}_{i}(\cdot)$ are globally exponentially stable such that for all $i$ and for every pair of switching times $(\tau_{i}^{u}, \tau_{i}^{a})$, the trajectories of the observer error satisfies the bound
\begin{equation}\label{eq:finalBound}
         \left\|\tilde{\theta}_{i}(\tau_{i+p}^{u})\right\| \leq \left\|\tilde{\theta}_{i}(\tau_{i}^{a})\right\|\prod_{j=0}^{p}e^{-\beta\Delta \tau_{i+j}^{a}}.
     \end{equation}
\end{theorem}
\begin{proof}
    Taking the Lie derivative of \eqref{eq:observerLyap} along the flow of \eqref{eq:observerErrorDynamics}, and substituting the observer error dynamics for the case where feature $\sigma_{i}(t) = u$ implies that $\dot{W}_{i}(\tilde{\theta}_{i}(t)) \leq 0$. Using the Comparison Lemma \cite[Lemma~3.4]{SCC.Khalil2002}, a bound on the observer error is obtained as
    \begin{equation}\label{eq:bound1}
        \left\|\tilde{\theta}_{i}(t)\right\| \leq \left\|\tilde{\theta}_{i}(\tau_{i}^{u})\right\|,
    \end{equation}
    while $\sigma_{i}(t) = u$.
    When $\sigma_{i}(t) = a$, $\lambda_{\min}\{\mathscr{Y}_{i}^{\top}(t)\mathscr{Y}_{i}(t)\} > \lambda_{a}$ implies the bound $\dot{W}_{i}(\tilde{\theta}_{i}(t)) \leq 
    -2\beta W_{i}(\tilde{\theta}_{i}(t))$, where $\beta = k_{\theta}\lambda_{a}$. Using the property of projection operators in \cite[Lemma E.1. IV]{SCC.Krstic.Kanellakopoulos.ea1995} and invoking \cite[Theorem~4.10]{SCC.Khalil2002}, it can be concluded that the observer error system is globally exponentially stable and by the Comparison Lemma \cite[Lemma~3.4]{SCC.Khalil2002}, the observer error is bounded as
    \begin{equation}\label{eq:bound2}
        \left\|\tilde{\theta}_{i}(t)\right\| \leq \left\|\tilde{\theta}_{i}(\tau_{i}^{a})\right\|e^{-\beta(t-\tau_{i}^{a})}.
    \end{equation}
    Substituting $t = \tau_{i}^{u}$ into the inequality in \eqref{eq:bound2} implies that $\|\tilde{\theta}_{i}(\tau_{i}^{u})\| \leq \|\tilde{\theta}_{i}(\tau_{i}^{a})\|e^{-\beta(\Delta \tau_{i}^{a})}$. Substituting
    the first inequality in \eqref{eq:bound1} into the previous inequality and extending this analysis recursively across $p$ intervals leads to the inequality in \eqref{eq:finalBound}.
\end{proof}
The following section establishes a minimum dwell-time condition, which ensures that features remain in the FOV of the camera long enough for distance estimation before leaving the FOV.

\section{Ensuring Stability Through Dwell-Time Conditions}

As established in Section~\ref{section:sequentialTracking}, using sequential feature tracking we can estimate the position of the deputy spacecraft relative to the goal $\hat{\vec{p}}_{g/b}(t)$, coordinatized in $\mathcal{F}_{H}$, using the distance estimates $\hat{\theta}_{i}(t)$ from the observer in \eqref{eq:updateLaw}. Each feature $i \in \mathcal{V}(t)$ provides a distance estimate during the finite interval $[\tau_{i}^{a}, \tau_{i}^{u})$ while it remains in the FOV of the camera. During this time, the estimation error $\tilde{\theta}_{i}(t)$ decays under the adaptive observer update law in \eqref{eq:updateLaw}. After the feature leaves the FOV, the corresponding estimate remains fixed, and a new feature is used to generate a new estimate. To ensure convergence of the error $\tilde{\vec{p}}_{g/b}(t)$ to a desired $\epsilon$-error bound, it is necessary that each new distance estimate $\hat{\theta}_{i}(t)$ used in the chaining process improves upon the previous one. Specifically, the estimation error at the end of each tracking interval of each feature $i \in \mathcal{V}(t)$ must be strictly less than that of the preceding feature. This paper formalizes this requirement using a minimum feedback availability dwell-time condition.
\begin{corollary}\label{cor:DwellTime}
Under the hypothesis of Theorem~\ref{thm:expObserver}, for each feature $i \in \mathcal{V}(t)$, the errors of the switched system defined by the switching signal $\sigma_{i}(t)$ and the observer update law in \eqref{eq:updateLaw} ensure that the estimation error $\tilde{\theta}_{i}(t)$ at time $t = \tau_{i}^{u}$ is bounded as
\begin{equation}\label{eq:dwellTimeCondition1}
\left\| \tilde{\theta}_{i}(\tau_{i}^{u}) \right\| \leq \overline{\theta}_{i-1}^{u} - \delta,
\end{equation}\label{eq:dwellTimeCond}
provided the switching signal $\sigma_{i}(t)$ satisfies the minimum feedback availability dwell-time condition
\begin{equation}
\Delta \tau_{i}^{a} \geq -\frac{1}{\beta} \ln\left( \frac{\overline{\theta}_{i-1}^{u} - \delta}{\overline{\theta}_{i}^{a}} \right),
\end{equation}
where $\delta \geq 0$ is the desired margin of improvement from the previous distance estimate, $\overline{\theta}_{i-1}^{u} > 0$ is the bound on the observer error $\tilde{\theta}_{i-1}(\tau_{i-1}^{u})$ of the previous feature in the sequence such that $\|\tilde{\theta}_{i-1}(\tau_{i-1}^{u})\| \leq \overline{\theta}_{i-1}^{u}$, and $\overline{\theta}_{i}^{a} > 0$ is a constant that satisfies $\|\tilde{\theta}_{i}(\tau_{i}^{a})\| \leq \overline{\theta}_{i}^{a}$.
\end{corollary}
\begin{proof}
   For each $i \in \mathcal{V}(t)$, during the interval $t \in [\tau^{a}_{i}, \tau^{u}_{i})$ where feature $i$ is tracked, the observer error satisfies $\|\tilde{\theta}_{i}(\tau_{i}^{u})\| \leq \|\tilde{\theta}_{i}(\tau_{i}^{a})\|e^{-\beta(\Delta \tau_{i}^{a})}$ under the adaptive update law in \eqref{eq:updateLaw}. Given that the initial error is bounded as $\|\tilde{\theta}_{i}(\tau_{i}^{a})\| \leq \overline{\theta}_{i}^{a}$ for some constant $\overline{\theta}_{i}^{a} > 0$ and since we want to ensure that the final error $\|\tilde{\theta}_{i}(\tau_{i}^{u})\|$ improves upon the previous estimate $\|\tilde{\theta}_{i-1}(\tau_{i-1}^{u})\|$ by at least $\delta > 0$, it suffices to require that
\begin{equation}
\left\| \tilde{\theta}_{i}(\tau_{i}^{a}) \right\| e^{-\beta \Delta \tau_{i}^{a}} \leq \overline{\theta}_{i-1}^{u} - \delta.
\end{equation}
 Therefore, solving for $\Delta t_{i}^{a}$ yields the minimum feedback availability dwell-time condition in \eqref{eq:dwellTimeCondition1} for each $i \in \mathcal{V}(t)$.
\end{proof}
\begin{remark}
The dwell-time condition in \eqref{thm:expObserver} guarantees that each newly recorded distance estimate reduces the position estimation error $\tilde{\vec{p}}_{g/b}(t)$ when compared to the last estimate from a previously observed feature. Given that the $i$-th distance estimate remains fixed once the corresponding $i$-th feature exits the FOV, the reduction in error $\tilde{\vec{p}}_{g/b}(t)$ is achieved during the active tracking phase of a new feature. We transition to this new feature once the dwell time condition in \eqref{fig:dwellTime} is met, as illustrated in Fig.~\ref{fig:dwellTime} and used the newly observed feature to compute the next estimate of $\hat{\vec{p}}_{g/b}(t)$. Over time, this sequential refinement drives the position estimation error $\tilde{\vec{p}}_{g/b}(t)$ toward the target error $\epsilon$.
\end{remark}

The next section uses the estimates of $\hat{\vec{p}}_{g/b}(t)$ to develop an approximate Lagrangian-based policy that enables the deputy spacecraft to inspect all the points on the chief satellite.

\section{Safe Control Design under Illumination and FOV Constraints}\label{section:controlDesign}
\subsection{Observer Robust Control Barrier Functions}
To formulate the constrained optimal control problem, we define the state vector of the goal relative to the deputy spacecraft in the Hill frame $\mathcal{F}_{H}$ as $\mathbf{x}(t) \coloneqq \begin{bmatrix} \vec{p}_{g/b}^{\top}(t) & \vec{v}_{g/b}^{\top}(t) \end{bmatrix}^{\top} \in \mathbb{R}^{6}$.
The state evolves according to the linearized CW dynamics, given by
\begin{equation}\label{eq:linearDynamics}
    \dot{\mathbf{x}} = A\mathbf{x}(t) + B \mathbf{u}(t), \quad \mathbf{x}(0) = \mathbf{x}_{0},
\end{equation}
where 
$A = \begin{bmatrix}
 0 & 0 & 0 & 1 & 0 & 0 \\
 0 & 0 & 0 & 0 & 1 & 0 \\
 0 & 0 & 0 & 0 & 0 & 1 \\
 3n^{2} & 0 & 0 & 0 & 2n & 0 \\
 0 & 0 & 0 & -2n & 0 & 0 \\
 0 & 0 & -n^{2} & 0 & 0 & 0 
\end{bmatrix}$, $B =  \frac{1}{m}\begin{bmatrix}
 0_{3 \times 3} \\ 
 I_{3 \times 3}
\end{bmatrix}$, and $\mathbf{u}(t) = \begin{bmatrix} \left(F(t)\right)_{x} & \left(F(t)\right)_{y} & \left(F(t)\right)_{z} \end{bmatrix}^{\top} \in \mathbb{R}^{3}$. Under the assumption that the goal is stationary in the chief frame, the relative velocity satisfies  
$\vec{v}_{g/b}(t) = -\vec{v}_{b/h}(t)$. Because $\vec{p}_{g/b}$ is unknown and must be estimated via the MRE described in Section~\ref{section:MREObserver}, the controller will have to depend on state estimates, denoted $\hat{\mathbf{x}}_{i} \in \mathbb{R}^{3}$, instead of the true state of $\mathbf{x}(t)$, defined as $\hat{\mathbf{x}}(t) \coloneqq \begin{bmatrix}\hat{\vec{p}}_{g/b, h_{i}}^{\top}(t) & \vec{v}_{g/b}^{\top}(t)\end{bmatrix}^{\top} \in \mathbb{R}^{6}$, where $\hat{\vec{p}}_{g/b, h_{i}} \coloneqq \vec{p}_{g/h_{i}} - \vec{u}_{b/h_{i}}(t)\hat{r}_{b/h_{i}}(t)$. As highlighted in \cite{SCC.Xu.Tabuada.ea2015, SCC.Jankovic.ea2018,SCC.Kolathaya.Ames.ea2019,SCC.Agrawal.Panagou.ea2022}, guaranteeing safety using an approximate controller requires robustification of barrier functions. Motivated by the results in \cite{SCC.Agrawal.Panagou.ea2022}, this paper employs an observer-based robust control barrier function (ORCBF) to provide robustness to estimation errors. To this end, let the ORCBF $\Phi \in \mathbb{R}^{6} \times [0, \; t_{f}] \to \mathbb{R}$ be defined as 
\begin{equation}\label{eq:observerRobustBarrier}
     \Phi(\mathbf{x}, t) \coloneqq \left(b_{r}(\mathbf{x}, t) - b_{r}(0, t)\right)^{2},
 \end{equation}
 where $t_{f} \in \mathbb{R}_{\geq 0}$ is the final time for a fixed goal location, $b_{r}(\mathbf{x}, t) \coloneqq  \frac{\gamma_{\Phi}}{h_{r}(\mathbf{x}, t)}$  is a robustified barrier function and  $\gamma_{\Phi} > 0$ is the barrier penalty gain. The robustified constraint function is given by $h_{r}(\mathbf{x}, t) = h(\mathbf{x}) - L_{h} M(t)$, where the constraint function $h: \mathbb{R}^{6} \to \mathbb{R}$ is locally Lipschitz continuous on a compact set $\Omega \subset \mathbb{R}^{6}$ containing the origin, such that there exists a constant $L_{h} > 0$ that satisfies
\begin{equation}\label{eq:lipschitzH}
    \vert  h(\mathbf{x})- h(\hat{\mathbf{x}}_{i})\vert \leq  L_{h} \left\|\mathbf{x}-\hat{\mathbf{x}}_{i}\right\|,
\end{equation}
for all $\mathbf{x}, \hat{\mathbf{x}}_{i} \in \Omega$. The term $M: [0, \; t_{f}] \to \mathbb{R}$ is a non-increasing robustifying function, designed to compensate for state estimation errors, and is defined as 
\begin{equation}\label{eq:compensationTerm}
     M(t) \coloneqq \epsilon_{r}e^{-\beta t},
\end{equation}
where $\epsilon_{r} > 0$ is a known bound that satisfies $\left\|\mathbf{x}(0)-\hat{\mathbf{x}}_{i}(0)\right\| \leq \epsilon_{r}$ and can be computed as
\begin{equation}\label{eq:initialBound}
    \epsilon_{r} = \scalemath{0.92}{\max_{r \in [r_{\min}, r_{\max}]} \left\{\left\| r \vec{u}_{k/h_{i}}(0) - \vec{p}_{g/h}(0) - \hat{\vec{p}}_{g/b}(0) \right\|\right\}},
\end{equation}
using the relationships in \eqref{eq:goalPosition} and \eqref{eq:deputyKeyFrame}.
 The constraint function $h$ encodes keep-out-zone (KOZ) and keep-in-zone (KIZ) constraints. It satisfies $\nabla_{\mathbf{x}}h(\mathbf{x}) \neq 0$ for all $\mathbf{x} \in \{\mathbf{x} \in \Omega \mid h(\mathbf{x}) \geq 0\}$ and for all $t \in \left[0, \; t_{f}\right]$, and is defined as
\begin{equation}
    h(\mathbf{x}) \coloneqq \frac{h_{\textrm{KOZ}}(\mathbf{x}) \cdot h_{\mathrm{KIZ}}(\mathbf{x})}{h_{\textrm{KOZ}}(\mathbf{x})+h_{\mathrm{KIZ}}(\mathbf{x})}.
\end{equation}
The function $h_{\textrm{KOZ}}: \mathbb{R}^{6} \to \mathbb{R}_{\geq 0}$ enforces KOZ constraint, ensuring the deputy spacecraft does not collide with the chief satellite, and is given by \cite{SCC.Dunlap.Wijk.ea2023, SCC.Dunlap.Bennett.ea2024}
\begingroup\medmuskip=0mu\begin{equation}
    h_{\text{KOZ}}(\mathbf{x}) =\sqrt{2a_{\max}\left(\left\|\vec{p}_{g/b}+\vec{p}_{g/h}\right\| - r_{\min}\right)} + \vec{v}_{\mathrm{proj}}(\mathbf{x}),
\end{equation}\endgroup
where $a_{\max}$ is the maximum acceleration from natural motion and control limits, and $\vec{v}_{\mathrm{proj}}(\mathbf{x}) = \frac{\langle\vec{v}_{b/h}, \;\vec{p}_{b/h}\rangle}{\left\|\vec{p}_{b/h}\right\| }$. Conversely, the function $ h_{\textrm{KIZ}}: \mathbb{R}^{6} \to \mathbb{R}_{\geq 0}$ enforces the KIZ constraint, ensuring the deputy spacecraft does not stray too far from the chief satellite, and is given by \cite{SCC.Dunlap.Wijk.ea2023, SCC.Dunlap.Bennett.ea2024} 
 \begingroup\medmuskip=0mu\begin{equation}
     h_{\textrm{KIZ}}(\mathbf{x}) =\sqrt{2a_{\max}\left(r_{\max} - \left\|\vec{p}_{g/b}+\vec{p}_{g/h}\right\|\right)} + \vec{v}_{\mathrm{proj}}(\mathbf{x}).
\end{equation}\endgroup
\begin{figure}
    \centering
    \input{figures/avoidanceRegionsPlot}
    \caption{Avoidance Regions}
    \label{fig:avoidance}
\end{figure}
To facilitate the subsequent stability analysis in Section~\ref{section:CWstabilityAnalysis}, we first establish the positive definiteness of the ORCBF in \eqref{eq:observerRobustBarrier}, 
\begin{lemma}\label{lem:phibounds}
Let the ORCBF $(\mathbf{x}, t) \mapsto \Phi(\mathbf{x}, t)$ be defined as in \eqref{eq:observerRobustBarrier} and $\mathcal{X} \subseteq \Omega$ be the compact set of states $\mathbf{x}$, containing the origin, for which the safety constraints are strictly satisfied, i.e., $h(\mathbf{x}) > 0$ for all $\mathbf{x} \in \mathcal{X}$. If $\mathbf{x} \mapsto h(\mathbf{x})$ is continuously differentiable and satisfies $h(\mathbf{x}) - L_{h} M(t) > 0$ for all $\mathbf{x} \in \mathcal{X}$ and $t \geq [0, \; t_{f}]$, and $b_{r}$ is locally co-Lipschitz in $\mathbf{x}$, uniformly in $t$, then $\Phi$ is positive definite (PD) on $\mathcal{X} \times [0, t_{f}]$ containing the origin. Furthermore, there exist class $\mathcal{K}$ functions $\alpha_1(\cdot)$ and $\alpha_2(\cdot)$ such that
$$ \alpha_{1}\left(\left|h(\mathbf{x}) - h(\mathbf{x}(0))\right|\right) \leq \Phi(\mathbf{x},t) \leq \alpha_{2}\left(\left|h(\mathbf{x}) - h(\mathbf{x}(0))\right|\right), $$
for all $\mathbf{x} \in \mathcal{X}$ and $t \in [0, \; t_{f}]$.
\end{lemma}

\begin{proof}
First, we establish the positive definiteness of $\Phi$ on the set $\mathcal{X} \times [0, t_{f}]$. By assumption, $h$ is continuously differentiable and thus continuous. For all $\mathbf{x} \in \mathcal{X} $, we are given that $h(\mathbf{x}) > 0$ and $h(\mathbf{x}) > L_{h}M(t)$, so both $b_r(\mathbf{x}, t)$ and $b_r(0, t)$ are finite and well-defined. Clearly, from \eqref{eq:boundedORCBF}, $\Phi(0, t) = 0$ for all $t \in [0, t_{f}]$, and $\Phi(\mathbf{x}, t) > 0$ for all $\mathbf{x} \in \mathcal{X}\setminus \{0\}$. Since $b_{r}$ is locally co-Lipschitz in $\mathbf{x}$, uniformly in $t$, then there exist a constant $L_{b} > 0$ such that
\begin{equation}
    \vert b_{r}(\mathbf{x}, t) - b(0, t) \vert \geq L_{b}\|\mathbf{x}\|, \quad \forall \mathbf{x} \in \mathcal{X}, t \in [0, t_{f}],
\end{equation}
which implies $\Phi(\mathbf{x}, t) = \left(b_{r}(\mathbf{x}, t) - b_{r}(0, t)\right)^{2} \geq \alpha_{3}(\|x\|)$, where $\alpha_{3}(\|x\|) \coloneqq L_{b}^{2}\|\mathbf{x}\|^{2}$. Thus, $\Phi$ is PD on $\mathcal{X} \times [0, t_{f}]$.

To establish the $\mathcal{K}$-function bounds, let $y = h(\mathbf{x})$. From the definition of the ORCBF in \eqref{eq:observerRobustBarrier}, we can write 
\begin{equation}
    \Phi(y, t) = \left( \frac{\gamma_{\Phi}}{y - L_{h} M(t)} - \frac{\gamma_{\Phi}}{h(\mathbf{x}(0)) - L_{h} M(t)} \right)^{2}.
\end{equation}
Factoring out $\gamma_{\Phi}$ and combining terms, we obtain
\begin{multline}
   \Phi(y, t) = \gamma_{\Phi}^{2} \left( \frac{(h(\mathbf{x}(0)) - L_{h} M(t)) - (y - L_{h} M(t))}{(y - L_{h} M(t))(h(\mathbf{x}(0)) - L_{h} M(t))} \right)^{2} \\
 = \gamma_{\Phi}^{2} \frac{(h(\mathbf{x}(0)) - y)^{2}}{(y - L_{h} M(t))^{2} (h(\mathbf{x}(0)) - L_{h} M(t))^{2}} 
\end{multline}
Let $\epsilon_{y} = |y - h(\mathbf{x}(0))| = |h(\mathbf{x}) - h(\mathbf{x}(0))|$. Then $(h(\mathbf{x}(0)) - y)^{2} = \epsilon_{y}^{2}$.
So,
\begin{equation}\label{eq:boundedORCBF}
\Phi(y, t) = \gamma_{\Phi}^{2} \frac{\epsilon_{y}^{2}}{(y - L_{h} M(t))^{2} (h(\mathbf{x}(0)) - L_{h} M(t))^{2}}.
\end{equation}
Note that the robustifying term $M(t)$ (defined in \eqref{eq:compensationTerm}) is continuous and bounded over $[0, \; t_{f}]$. Thus, the denominator in \eqref{eq:boundedORCBF} is strictly positive and bounded over the compact set $\mathcal{X} \times [0, t_f]$, there exist constants $d_{\min}, d_{\max} > 0$ such that
\begin{equation}
\frac{\gamma_{\Phi}^{2}}{d_{\max}} \epsilon_{y}^{2} \leq \Phi(y, t) \leq \frac{\gamma_{\Phi}^{2}}{d_{\min}} \epsilon_{y}^{2}.
\end{equation}
Letting $\alpha_1(\epsilon_{y}) = \frac{\gamma_{\Phi}^{2}}{d_{\max}} \epsilon_{y}^{2}$ and $\alpha_2(\epsilon_{y}) = \frac{\gamma_{\Phi}^{2}}{d_{\min}} \epsilon_{y}^{2}$, the result follows.
\end{proof}

The following definition formalizes the notion of safety considered in this paper.
\begin{definition}\label{defn:safety}
    The system in \eqref{eq:linearDynamics} is safe with respect to the sets $\left(\mathcal{X}_{0}, \mathcal{X}\right)$ if $\mathbf{x}_{0} \in \mathcal{X}_{0}$ implies $\mathbf{x}(t) \in \mathcal{X}$ for all $t \in \mathbb{R}_{\geq 0}$.
\end{definition}
The subsequent theorem, inspired by \cite{SCC.Agrawal.Panagou.ea2022}, establishes that the ORCBF in \eqref{eq:observerRobustBarrier} guarantees the safety of the system in \eqref{eq:linearDynamics} with respect to the sets $(\mathcal{X}_{0}, \mathcal{X})$, despite the presence of state estimation errors.
\begin{theorem}\label{thm:CBFsafety}
\cite{SCC.Agrawal.Panagou.ea2022} For the system dynamics given by \eqref{eq:linearDynamics} and the ORCBF defined in \eqref{eq:observerRobustBarrier}, let $\mathcal{K}_{\operatorname{safe}}(\mathbf{x}, t)$ be the set of control inputs defined as
\begin{equation}\label{eq:controlSetU}
    \mathcal{K}_{\operatorname{safe}}(\mathbf{x}, t) \coloneqq \left\{ \mathbf{u} \in \mathbb{R}^{3} \mid c(\mathbf{x}, \mathbf{u}, t) \leq 0 \right\},
\end{equation}
with $c(\mathbf{x}, \mathbf{u}, t) \coloneqq \nabla_{\mathbf{x}} \Phi(\mathbf{x},t) \left( A \mathbf{x} + B \mathbf{u} \right) + \nabla_{t} \Phi(\mathbf{x},t) -  \alpha(h_{r}(\mathbf{x},t))$, where $\alpha(\cdot)$ is a class $\mathcal{K}$ function. If the initial states are such that $\mathbf{x}(0) \in \mathcal{X}_{0} \subset \mathcal{X}$, the initial condition of the observer in \eqref{eq:updateLaw} satisfies $\hat{\mathbf{x}}_{i}(0) \in \hat{\mathcal{X}}_{0} \coloneqq \{\hat{\mathbf{x}}_{i} \in \mathbb{R}^{6} \mid h_r(\hat{\mathbf{x}}_{i}, 0) \geq 0\}$, and the estimates of the relative distances are updated using the observer in \eqref{eq:updateLaw} subject to the dwell-time condition in \eqref{fig:dwellTime}, then any Lipschitz continuous controller $\mathbf{u}:[0, t_{f}] \to \mathbb{R}^{3}$ that satisfies $\mathbf{u}(t) \in \mathcal{K}_{\operatorname{safe}}(\mathbf{x}, t)$ ensures that the system is safe with respect to the sets $(\mathcal{X}_{0}, \mathcal{X})$.
\end{theorem}
 \begin{remark}
      It should be noted that the stability results presented in \cite{SCC.Agrawal.Panagou.ea2022} are not directly utilized in the subsequent control design; they only justify the choice of the ORCBF in \eqref{eq:observerRobustBarrier}.
 \end{remark}
The following assumption is necessary to facilitate the development of the approximate Lagrangian-based control policy in the next section.
 \begin{assumption}\label{ass:safetyControllability}
The control effectiveness matrix $B$ and the ORCBF $\Phi$ satisfy $\left\| \nabla_{\mathbf{x}} \Phi(\mathbf{x}, t) B \right\| \neq 0
\quad \forall\, \mathbf{x} \in \operatorname{Int}(\mathcal{X}) \setminus \{0\},\;
\forall\, t \in [0, t_f]$.
\end{assumption}
 \begin{remark}
Assumption~\ref{ass:safetyControllability} ensures that the ORCBF \eqref{eq:boundedORCBF} has a well-defined gradient direction aligned with the control input \( u(t) \in \mathcal{K}_{\operatorname{safe}}(\mathbf{x}, t) \) throughout the interior of the safe set \( \mathcal{X} \). This condition, analogous to controllability assumptions in \cite{arXivSCC.Bandyopadhyay.Bhasin.ea2023} and \cite{SCC.Cohen.Belta.ea2023}, guarantees the existence of admissible control inputs that guarantee safety of the system \eqref{eq:CLdynamics} with respect to the sets \( (\mathcal{X}_{0}, \mathcal{X}) \).
\end{remark}

\subsection{Information Maximizing Constrained Safe Control Using Robust Barrier Functions}
  The control objective is to find an admissible control policy $\mathbf{u}:[0,  \; t_{f}) \to \mathbb{R}^{3}$ that solves the finite horizon minimization problem
\begin{multline}\label{eq:coop}
\begin{aligned}
    & \min_{\mathbf{u} \in \mathbb{R}^{3}} \quad \int_{0}^{t_{f}} r\left(\mathbf{x}(\tau), \mathbf{u}(\tau)\right) \, \mathrm{d}\tau + V(\mathbf{x}(t_{f}), t_{f}), \\
    & \text{s.t.} \quad \mathbf{u}(t) \in \mathcal{K}_{\operatorname{safe}}(\mathbf{x}(t), t), \quad \forall t \in [0, \; t_{f}],
\end{aligned}
\end{multline}
 subject to the dynamics in \eqref{eq:linearDynamics}, where $V: \mathbb{R}^{6} \times \left[0, \; t_{f}\right] \to \mathbb{R}_{\geq 0}$ is the terminal cost, selected as
\begin{equation}\label{eq:terminalCost}
     V(\mathbf{x}(t_{f}), t_{f}) = \mathbf{x}^{\top}(t_{f})Q_{f}\mathbf{x}(t_{f}),
 \end{equation}
 with $\nabla_{\mathbf{x}(t_{f})}V(\mathbf{x}(t_{f})) =  2\mathbf{x}^{\top}(t_{f})Q_{f}$, where $Q_{f} = Q_{f}^{\top} \succ 0 \in \mathbb{R}^{6 \times 6}$ is a user-selected terminal state penalty matrix. The running cost $r: \mathbb{R}^{6} \times \mathbb{R}^{3} \to \mathbb{R}$ in \eqref{eq:coop} is defined as
\begin{equation}\label{eq:runningCost}
    r(\mathbf{x}, \mathbf{u}) \coloneqq  \mathbf{x}^{\top} Q \mathbf{x} + \mathbf{u}^{\top}R\mathbf{u} + \gamma_{c}\left\langle C\mathbf{x}, \vec{n}_{c}\right\rangle^{2},
 \end{equation}
 where $Q \in \mathbb{R}^{6 \times 6}$ and $R \in \mathbb{R}^{3 \times 3}$ are the state and control penalty matrices, respectively, satisfying $Q = Q^{\top} \succ 0$,  $R = R^{\top} \succ 0$, $C \coloneqq \begin{bmatrix} 0_{3 \times 3} && I_{3 \times 3}
 \end{bmatrix} \in \mathbb{R}^{3 \times 6}$, $\gamma_{c} \in \mathbb{R}_{> 0}$ is a user-defined constant, and $\vec{n}_{c}\in \mathbb{R}^{3}$ denotes the outward unit normal to the plane defined by the image-plane locations of the illuminated features currently tracked by the camera. 
 
 The penalty $\gamma_{c}\langle C\mathbf{x}, \vec{n}_{c}\rangle^{2}$ in \eqref{eq:runningCost} enforces motion orthogonal to the feature plane, which improves the conditioning of the regressor in \eqref{eq:rankCond} and thereby enhances the numerical stability of the MRE-based observer in \eqref{eq:updateLaw}, as shown in \cite{SCC.Ogri.Qureshi.ea2024}.  
Equivalently, \eqref{eq:runningCost} can be written as $r(\mathbf{x}, \mathbf{u}) = \mathbf{x}^{\top} \underline{Q} \mathbf{x} + \mathbf{u}^{\top}R\mathbf{u}$,
where $\underline{Q} = \underline{Q}^{\top} \succ 0$ is defined as $\underline{Q} \coloneqq Q + \gamma_{c}C^{\top}\vec{n}_{c}\vec{n}_{c}^{\top}C$.

\begin{assumption}\label{ass:sufficient_time}
The final time $t_{f} > 0$ is sufficiently large such that there exists at least one admissible control input $\mathbf{u}(t) \in \mathcal{K}_{\operatorname{safe}}(\mathbf{x}, t)$ for all $t \in [0, t_{f}]$, under which the relative position between the goal and the deputy satisfies $\|\vec{p}_{g/b}(t_{f})\| = 0$, for a given goal location i.e., the goal location is reached at time $t_{f}$. 
\end{assumption}


\subsection{ Lagrangian-Based Feedback Control Policy}
In this section, a Lagrangian-based feedback control policy is developed. Let  $V^{*}: \mathbb{R}^{6} \times [0, \, t_{f}]  \rightarrow \mathbb{R}$ denote the optimal value function which is assumed to be continuously differentiable with $V^*(0, t_{f}) = 0$. Based on the constrained optimization problem in \eqref{eq:coop}, we formulate} the following constrained optimization problem (cf. \cite[Equation~27]{SCC.Nevistic.Primbs1996}),
\begingroup\medmuskip=1mu\begin{equation}\label{eq:HJB} 
\begin{aligned}
 0 =  \nabla_{t}V^{*}(\mathbf{x}, t) + & \min_{\mathbf{u} \in \mathcal{U}} \Big\{ \nabla_{\mathbf{x}}V^{*}(\mathbf{x}, t)\left( A\mathbf{x} + B\mathbf{u} \right)  + r(\mathbf{x}, \mathbf{u})\Big\}, \\
& \textrm{s.t.} \quad c(\mathbf{x}, \mathbf{u}, t) \leq  0,
\end{aligned}
\end{equation}\endgroup
where the constraint $c$ is introduced in \eqref{eq:controlSetU} to ensure that $\mathbf{u}(t) \in \mathcal{K}_{\operatorname{safe}}(\mathbf{x}, t)$.\footnote{While the constrained optimization problem in \eqref{eq:HJB} facilitates online synthesis of a safe controller, it is not generally equivalent to the constrained optimal control problem in \eqref{eq:coop}. Examination of conditions under which problems \eqref{eq:coop} and \eqref{eq:HJB} generate equivalent optimal controllers is out of the scope of this paper.}
To synthesize a safe constrained controller, we define a Lagrangian for the constrained optimization problem in \eqref{eq:HJB} as
\begin{multline}\label{eq:lagrangian}
    \mathcal{L}(\mathbf{x}, \mathbf{u}, \lambda, t) =   \nabla_{\mathbf{x}}V^{*}(\mathbf{x}, t)\left( A\mathbf{x} + B \mathbf{u} \right) \\  \qquad  + \mathbf{x}^{\top} \underline{Q} \mathbf{x} + \mathbf{u}^{\top}R\mathbf{u} + \lambda c(\mathbf{x}, \mathbf{u}, t),
\end{multline}
for all $\mathbf{x} \in \mathbb{R}^{6}$ and $t \in [0, t_{f})$, where $\lambda$ denotes the Lagrange multiplier.
The Lagrangian-based control policy can be obtained by solving the Karush–Kuhn–Tucker conditions \cite{SCC.Bryson.Ho1975, SCC.Nevistic.Primbs1996} for the constrained optimization problem, which provides the following necessary conditions for optimality
\begin{align}
    &2\mathbf{u}^{\top}R {+} \left(\nabla_{\mathbf{x}}V(\mathbf{x}, t) {+} \lambda\nabla_{\mathbf{x}} \Phi(\mathbf{x},t)\right)B = 0,\label{eq:stationarity}\\
    &\lambda c(\mathbf{x}, \mathbf{u}, t) = 0, \label{eq:slackness} \\
    &\lambda  \geq 0.\label{eq:dual}
\end{align}
By substituting $\mathbf{u}(t)$ in \eqref{eq:stationarity} into the complementary slackness condition in \eqref{eq:slackness} and using the condition on the dual constraint in \eqref{eq:dual}, the optimal Lagrange multiplier for the optimal solution of \eqref{eq:HJB}, denoted by $\lambda^{*}: \mathbb{R}^{6} \times [0, t_{f}] \to \mathbb{R}$, can be expressed as (cf .\cite{SCC.Sforni.Notarstefano.ea2024})
\begin{equation}\label{eq:optimalLagrange}
\lambda^{*}(\mathbf{x}, t) = 
\frac{
\begin{aligned}
    &2\Big(\nabla_{\mathbf{x}} \Phi(\mathbf{x},t)A\mathbf{x} + \nabla_{t} \Phi(\mathbf{x},t) - \alpha(h_{r}(\mathbf{x},t))
    \\
    &\quad - \frac{1}{2} \nabla_{\mathbf{x}} \Phi(\mathbf{x},t)BR^{-1}B^{\top} \nabla_{\mathbf{x}} (V^{*}(\mathbf{x}, t))^{\top}  \Big)
\end{aligned}
}{
\nabla_{\mathbf{x}}\Phi(\mathbf{x},t)BR^{-1}B^{\top}\nabla_{\mathbf{x}} \Phi^{\top}(\mathbf{x}, t)}
\end{equation}
if $c(\mathbf{x}, \mathbf{u}^{*}(\mathbf{x}, t), t) > 0$ and $\lambda^{*}(\mathbf{x}, t) = 0$ otherwise. By solving the stationarity condition in \eqref{eq:stationarity}, the Lagrangian-based controller is given by the feedback policy $\mathbf{u}^{*}: \mathbb{R}^{6} \times\left[0, \; t_{f}\right]  \to \mathbb{R}^{3}$ which can be computed explicitly as
\begin{multline}\label{eq:optimalU}
    \mathbf{u}^{*}(\mathbf{x}, t) = -\frac{1}{2}R^{-1}B^{\top} \nabla_{\mathbf{x}}\left(V^{*}(\mathbf{x}, t)\right)^{\top} \\-\frac{1}{2}R^{-1}B^{\top}\lambda^{*}(\mathbf{x}, t)\nabla_{\mathbf{x}}\Phi^{\top}(\mathbf{x}, t).
\end{multline}
   Substituting \eqref{eq:terminalCost} and \eqref{eq:optimalU} into the unconstrained form of the HJB in \eqref{eq:HJB} yields the associated HJB equation
\begingroup\medmuskip=1mu\begin{multline}\label{eq:HJB2}
   \mathbf{x}^{\top}\scalemath{0.97}{\left(\dot{P}(t) + A^\top P(t) + P(t)A - P(t)BR^{-1}B^{\top}P(t) + \underline{Q} \right)}\mathbf{x}\\ = -\eta(\mathbf{x},t)^{\top}R\eta(\mathbf{x},  t)  \leq  0,
\end{multline}\endgroup
where $\eta(\mathbf{x}, t) = \frac{1}{2}R^{-1}B\lambda^{*}\nabla_{\mathbf{x}}\Phi^{\top}(\mathbf{x}, t)$ and
$P(t)$ is the solution of the Riccati differential equation corresponding to \eqref{eq:linearDynamics}, expressed as, (cf. \cite[Chapter 8, Section 8.2.1]{SCC.Tedrake2023})
 \begin{equation}\label{eq:ARE}
    \dot{P}(t) = -A^{\top} P(t) - P(t) A + P(t) BR^{-1}B^{\top} P(t)  - \underline{Q},
\end{equation}
 with terminal condition $P(t_{f}) = Q_{f}$, where $\underline{Q}_{f} = Q_{f} + \gamma_{c}C^{\top}\vec{n}_{c}\vec{n}_{c}^{\top}C$ and $\underline{Q}_{f} = \underline{Q}_{f}^{\top} \succ 0$. For all $t \in [0, \, t_{f}]$, the solution $P(t)$ satisfies the bounds $\underline{p}\, I_{6 \times 6} \leq P(t) \leq \overline{p}\, I_{6 \times 6}$, where $\underline{p}$, and $\overline{p} > 0$ are constant bounds \cite{SCC.Bonnabel.Slotline.ea2015}. By design, the terminal cost in \eqref{eq:terminalCost} is selected as the value function $V^*$ associated with the unconstrained optimal control problem.
 
 Using the ORCBF in \eqref{eq:observerRobustBarrier} and distance estimates from the distance observer in \eqref{eq:updateLaw}, the approximate Lagrangian-based control policy $\hat{\mathbf{u}}: \mathbb{R}^{6} \times\left[0, \; t_{f}\right]  \to \mathbb{R}^{3}$ is designed as
\begin{equation}\label{eq:uControl}
\hat{\mathbf{u}}(\hat{\mathbf{x}}_{i}, t) = -R^{-1}B^{\top}P\hat{\mathbf{x}}_{i} -\frac{1}{2}R^{-1}B^{\top}\hat{\lambda}\nabla_{\hat{\mathbf{x}}_{i}}\Phi^{\top}(\hat{\mathbf{x}}_{i}, t),
\end{equation}
where $\hat{\lambda}: \mathbb{R}^{6} \times [0, \; t_{f}] \to \mathbb{R}$ is the approximate Lagrange multiplier is given by
\begin{equation}\label{eq:approxLagrange}
\hat{\lambda}(\hat{\mathbf{x}}_{i}, t) = 
\frac{
\begin{aligned}
    &2\Big(\nabla_{\hat{\mathbf{x}}_{i}} \Phi(\hat{\mathbf{x}}_{i},t)A\hat{\mathbf{x}}_{i} + \nabla_{t} \Phi(\hat{\mathbf{x}}_{i},t) - \alpha(h_{r}(\hat{\mathbf{x}}_{i},t))
    \\
    &\quad - \nabla_{\hat{\mathbf{x}}_{i}} \Phi(\hat{\mathbf{x}}_{i},t)BR^{-1}B^{\top} P(t)\hat{\mathbf{x}}_{i} \Big)
\end{aligned}
}{
\nabla_{\hat{\mathbf{x}}_{i}}\Phi(\hat{\mathbf{x}}_{i},t)BR^{-1}B^{\top}\nabla_{\hat{\mathbf{x}}_{i}} \Phi^{\top}(\hat{\mathbf{x}}_{i}, t)}
\end{equation}
if $c(\hat{\mathbf{x}}_{i}, \hat{\mathbf{u}}(\hat{\mathbf{x}}_{i}, t), t) > 0$ and $\hat{\lambda}(\hat{\mathbf{x}}_{i}, t) = 0$ otherwise. Unlike \cite{SCC.Bandyopadhyay.Bhasin.ea2023a}, our approach leverages the ORCBF framework to ensure safety despite state estimation errors, retaining the ability to turn off the Lagrange multiplier if the safety constraints are satisfied. Even though we approximate the Lagrange multiplier, the ORCBF used in \eqref{eq:approxLagrange} ensures safety under Theorem~\ref {thm:CBFsafety}.
The ORCBF-based controller policy is designed as
\begin{equation}\label{eq:controlLaw}
     \mathbf{u}(t) = \hat{\mathbf{u}}(\hat{\mathbf{x}}_{i}, t).
 \end{equation}
 to guarantee safety despite the presence of state estimation errors. The following section analyzes the convergence properties of the closed-loop system.

\section{Analysis of Spacecraft closed-loop System}\label{section:CWstabilityAnalysis}
In this section, we analyze the system \eqref{eq:linearDynamics} under the feedback control law \eqref{eq:uControl} using a Lyapunov-based approach. We establish that the closed-loop trajectories remain uniformly bounded and that the system \eqref{eq:linearDynamics} is safe with respect to the sets $(\mathcal{X}_{0}, \mathcal{X})$. The following theorem presents the main theoretical results of this paper.
\begin{theorem}\label{thm:boundedSubsystem}
      Let $\chi > 0$ be such that $\bar{B}(0, \chi)\subset \mathcal{X} \times \mathbb{R}^{3}$. If Assumptions~\ref{ass:unObstructedRays}--\ref{ass:safetyControllability} hold and all the hypothesis of Lemma~\ref{lem:phibounds} are satisfied, the initial states are such that $\mathbf{x}(0) \in \mathcal{X}_{0}$ and the initial condition of the observer in \eqref{eq:updateLaw} satisfies $\hat{\mathbf{x}}_{i}(0) \in \hat{\mathcal{X}}_{0}$, the sufficient conditions in \eqref{eq:gainCondition} and \eqref{eq:secondCond} are satisfied, and the distance estimates are updated using the MRE-based observer in \eqref{eq:updateLaw}, then the trajectories of the system in \eqref{eq:linearDynamics} are uniformly bounded under the designed control policy \eqref{eq:uControl}.
\end{theorem}
\begin{proof}
To facilitate the analysis, let $Z_{i} \coloneqq \begin{bmatrix}
    \mathbf{x}^{\top} & \tilde{\theta}_{i}
\end{bmatrix} \in \mathcal{Z} \coloneqq  \mathcal{X} \times \mathbb{R}^{3}$ denote the state of the closed loop system. We start the proof with the analysis of the $i$-th subsystem which corresponds to the $i$-th feature being tracked, then develop a bound on $\|Z_{i}\|$ and show that if the trajectories of the closed-loop system do not escape to infinity for the $i$-th subsystem, then this holds for all subsequent subsystems and thus the bound holds for the entire switched system.

Let the candidate Lyapunov function for the closed-loop system, $V_{L}: \mathcal{Z} \times \left[0, \; t_{f}\right] \to \mathbb{R}$, be defined as
\begin{equation}\label{eq:distLyapFunc}
    V_{L}(Z_{i}, t) \coloneqq \mathbf{x}^{\top} P(t)\mathbf{x} + \Phi(\mathbf{x}, t) + \frac{1}{2}\tilde{\theta}_{i}^{\top}\tilde{\theta}_{i},
\end{equation}
for $i \in \mathcal{S}(t)$.  As shown in Lemma~\ref{lem:phibounds}, the ORCBF $\Phi$ is PD on the set $\mathcal{X} \times [0, t_{f}]$. Furthermore, since $P(t)$ is uniformly PD, the candidate Lyapunov function $V_{L}$ is PD and decrescent on $\mathcal{Z}$, i.e., $V_{L}$ can be bounded as
\begin{equation}\label{eq:lyapunovBound}
\underline{v}_{l}(\left\| Z_{i}\right\|)\leq V_{L}\left(Z_{i}, t\right)\leq \overline{v}_{l}(\left\|Z_{i}\right\|),
\end{equation}
for all $Z\in \mathcal{Z}$ and $t \in\left[0, \; t_{f}\right]$, where $\underline{v}_{l}(\cdot)$ and $\overline{v}_{l}(\cdot)$ are class $\mathcal{K}$ functions. By assumption, the ORCBF $\Phi$ is locally Lipschitz continuous on $\mathcal{X}$, uniformly in $t$, thus there exists a Lipschitz constant $L_{\Phi} > 0$, independent of $t$, such that 
\begin{equation}\label{eq:lipschitzBoundORCBF}
    \vert\Phi(\mathbf{x},t) - \Phi(\hat{\mathbf{x}}_{i}, t)\vert \leq L_{\Phi}\|\mathbf{x} - \hat{\mathbf{x}}_{i}\|, \quad \forall \mathbf{x}, \hat{\mathbf{x}}_{i} \in \mathcal{X}.
\end{equation}
Similarly, by assumption, the Lagrangian multiplier in \eqref{eq:optimalLagrange} is locally Lipschitz continuous on $\mathcal{X}$, uniformly in $t$, thus exists a Lipschitz constant $ > 0$, independent of $t$, such that 
\begin{equation}\label{eq:lipschitzBoundLambda}
   |\lambda^{*}(\mathbf{x}, t) - \hat{\lambda}(\hat{\mathbf{x}}_{i}, t)| \leq L_{\lambda}\|\mathbf{x}-\hat{\mathbf{x}}_{i}\|, \forall \mathbf{x}, \hat{\mathbf{x}}_{i} \in \mathcal{X}.
\end{equation}
Taking the Lie derivative of $V_{L}$ along the flow of \eqref{eq:linearDynamics} and \eqref{eq:observerErrorDynamics}, then substituting the control law in \eqref{eq:controlLaw} yields
\begin{multline}\label{eq:lyapDeriv1}
    \dot{V}_{L}(Z_{i}, t) = \mathbf{x}^{\top}(t)\bigl(\dot{P}(t) + A^\top P(t) + P(t)A \bigr)\mathbf{x}\\
    - \mathbf{x}^{\top}(t)P(t)BR^{-1}B^{\top}\hat{\lambda}(\hat{\mathbf{x}}_{i}, t)\nabla_{\hat{\mathbf{x}}_{i}}\Phi^{\top}(\hat{\mathbf{x}}_{i}, t)\\- \mathbf{x}^{\top}(t)P(t)BR^{-1}B^{\top}P(t)\hat{\mathbf{x}}_{i} + \nabla_{\mathbf{x}}\Phi(\mathbf{x}, t)A\mathbf{x} \\ - \nabla_{\mathbf{x}}\Phi(\mathbf{x}, t)BR^{-1}B^{\top}\hat{\lambda}(\hat{\mathbf{x}}_{i}, t)\nabla_{\hat{\mathbf{x}}_{i}}\Phi^{\top}(\hat{\mathbf{x}}_{i}, t) \\ - \nabla_{\mathbf{x}}\Phi(\mathbf{x}, t)BR^{-1}B^{\top}P(t)\mathbf{x} + \nabla_{t} \Phi(\mathbf{x},t)\\-\tilde{\theta}_{i}^{\top}\left[k_{\theta}\mathscr{Y}_{i}^{\top}(t)\mathscr{Y}_{i}(t)\right]\tilde{\theta}_{i}.
\end{multline}
 Note that $\|\mathbf{x} - \hat{\mathbf{x}}_{i}\| \leq \|\tilde{\theta}_{i}(t)\|$, since $\|\vec{p}_{g/b}(t) - \hat{\vec{p}}_{g/b}(t)\| \leq \|\vec{p}_{b/h_{i}}(t) - \hat{\vec{p}}_{b/h_{i}}(t)\|\leq \|r_{b/h_{i}}(t) - \hat{r}_{b/h_{i}}(t)\| \leq \|\tilde{\theta}_{i}(t)\|$. From \eqref{eq:goalPosition}, substituting \eqref{eq:ARE} into \eqref{eq:lyapDeriv1}, and using the bounds in \eqref{eq:lipschitzBoundORCBF} and \eqref{eq:lipschitzBoundLambda}, the Lie derivative of $V_L$ is bounded on $\mathcal{X} \times \mathbb{R}^3 \times [0, \; t_{f}]$ as 
\begin{multline}\label{eq:closedLoopLyapDerivBound}
    \dot{V}_{L}(Z_{i}, t) \leq -\lambda_{\min}(\underline{Q})\big\|\mathbf{x}\big\|^{2} - \beta\left\|\tilde{\theta}_{i}\right\|^{2} \\ -\varpi\lambda_{\max}(R_{B})\|\nabla_{\mathbf{x}}\Phi(\mathbf{x}, t)\|^{2} + \ell_{1}\|\mathbf{x}\|\|\tilde{\theta}_{i}\| \\+ \ell_{2}\|\mathbf{x}\| + \ell_{3}\|\tilde{\theta}_{i}\| + \ell_{4},
\end{multline}
 where $R_{B} \coloneqq BR^{-1}B^{\top}$ and the constants $\ell_{1}$, $\ell_{2}$, $\ell_{3}$, and $\ell_{4}$ are defined as $\ell_{1} \coloneqq \overline{p}^{2}\lambda_{\max}(R_{B}) + \overline{p}\lambda_{\max}(R_{B})\Upsilon$, $\ell_{2} \coloneqq \overline{p}\varpi \lambda_{\max}(R_{B})\overline{\nabla \Phi} + \overline{\nabla \Phi}\|A\|$, $\ell_{3} \coloneqq \varpi L_{\Phi}\lambda_{\max}(R_{B})\overline{\nabla\Phi} + \lambda_{\max}(R_{B})L_{\lambda}\overline{\nabla\Phi}^{2}$, and
$\ell_{4} \coloneqq \overline{\nabla_{t}\Phi}$, respectively, where $\Upsilon = \varpi L_{\Phi}+L_{\lambda}\overline{\nabla{\Phi}}$ and the constants $\varpi$, $L_{\Phi}$, $\overline{\nabla\Phi}$, and $\overline{\nabla_{t}\Phi}$ are positive and satisfy $\sup_{\mathbf{x}\in \mathcal{X}, t \in [0, \; t_{f}] }\|\lambda^{*}(\mathbf{x}, t)\| \leq \varpi$, $\sup_{\mathbf{x}\in \mathcal{X}}\|\nabla_{\mathbf{x}}\Phi(\mathbf{x}, t)\| \leq \overline{\nabla\Phi}$, and $\sup_{\mathbf{x}\in \mathcal{X}}\|\nabla_{t}\Phi(\mathbf{x}, t)\| \leq \overline{\nabla_{t}\Phi}$.
Applying completion of squares, and provided the gain condition \begin{equation}\label{eq:gainCondition}
       \frac{\lambda_{\min}(\underline{Q})}{3} > \frac{3\ell_{1}^{2}}{4\beta},
    \end{equation}
is satisfied, then the Lie derivative of $V_{L}$ can be bounded on $\mathcal{X} \times \mathbb{R}^3 \times [0, \; t_{f}]$ as
\begin{equation}
    \dot{V}_{L}(Z_{i}, t) \leq -\upsilon_{l}(\|Z_{i}\|) + \iota,
\end{equation}
 where $\iota = \frac{3\ell_{2}^{2}}{4\lambda_{\min}(\underline{Q})} + \frac{3\ell_{3}^{2}}{4\beta} + \ell_{4}$ is a constant and $\upsilon_{l}(\cdot)$ is a class $\mathcal{K}$ function that satisfies
$\upsilon_{l}(\|Z_{i}\|) \leq \frac{\lambda_{\min}(\underline{Q})}{3}\big\|\mathbf{x}\big\|^{2} - \frac{\beta}{3}\left\|\tilde{\theta}_{i}(t)\right\|^{2} -\varpi\lambda_{\max}(R_{B})\|\nabla_{\mathbf{x}}\Phi(\mathbf{x}, t)\|^{2}$. Letting $\rho \coloneqq \upsilon_{l}^{-1}(\iota)$, then, for all $\|Z_{i}\| > \rho$, it follows that $\dot{V}_{L}(Z_{i}, t) < 0$. If the sufficient condition 
\begin{equation}\label{eq:secondCond}
  \rho \leq  \overline{\upsilon}_{l}^{-1}\left(\underline{\upsilon}_{l}(\chi)\right),
\end{equation} is met, then \cite[Theorem~4.18]{SCC.Khalil2002} can be invoked to conclude that for trajectories starting from initial conditions satisfying $\|Z_{i}(0)\| \leq \overline{v}_{l}^{-1}(\underline{v}_{l}(\chi))$, there exist a $T \geq 0$ such that 
\begin{equation}\label{eq:stateBound}
     \|Z_{i}(t)\| \leq \underline{\upsilon}_{l}^{-1}(\overline{\upsilon}_{l}(\rho)), \quad \forall t \geq T.
\end{equation} 

From the bound in \eqref{eq:stateBound}, we know that the trajectories of the $i$-th subsystem do not blow up to infinity in finite time, which means that the result of Theorem~\ref{thm:expObserver} is applicable. Consequently, when switching from the $i$-th feature to the $(i+1)$-th feature, governed by the dwell time condition in \eqref{eq:dwellTimeCond}, we have that $\|Z_{i+1}(t)\| < \|Z_{i}(t)\|$ since $\mathbf{x}(t)$ is continuous across switches while the observer error $\tilde{\theta}_{i}(t)$ decays exponentially across switches, i.e., $\|\tilde{\theta}_{i+1}(t)\| < \|\tilde{\theta}_{i}(t)\|$, according to Theorem~\ref{thm:expObserver} and the switching logic described in Section~\ref{section:sequentialTracking}. Applying this argument recursively across all switching instances, we can conclude that the sequence of augmented state norms $\{Z_{i}(t)\}_{i=1}^{N}$ is strictly decreasing and bounded. This implies that the overall switched system state $Z(t)$ remains uniformly bounded for all $t \in [0, t_{f}]$, i.e., $\sup_{t\in [0, t_{f}]}\|Z(t)\| < \infty$.

Since $Z(t) \in \bar{B}(0, \chi) \subset \mathcal{X} \times \mathbb{R}^{3}$ and \(\mathbf{x}(t)\) is a component of $Z(t)$, then safety follows by Theorem~\ref{thm:CBFsafety}.

\end{proof}

\begin{figure}
    \centering
     \input{figures/spacecraft3DTrajectory1000}
    \caption{Trajectory of the deputy spacecraft over 1000 seconds while inspecting points on the chief satellite; inspected points are shown in green, and uninspected points in red}
    \label{fig:spacecraftTrajectory1000}
\end{figure}
\begin{figure}
    \centering
     \input{figures/spacecraft3DTrajectory3000}
    \caption{Trajectory of the deputy spacecraft over 3000 seconds while inspecting points on the chief satellite; inspected points are shown in green, and uninspected points in red}
    \label{fig:spacecraftTrajectory3000}
\end{figure}
\begin{figure}
    \centering
     \input{figures/PgbPlot}
    \caption{Trajectories of the error between the actual and estimated position of the goal location relative to the deputy spacecraft.}
    \label{fig:PgbTilde}
\end{figure}
\section{Simulation Results}\label{section:simulation}
In this section, we present a simulation study of a deputy spacecraft equipped with a monocular camera tasked with inspecting features on a chief. The chief satellite dynamics are modeled in Hill’s frame using the linearized CW equations \eqref{eq:linearDynamics}, with the chief satellite assumed to follow a circular orbit around Earth. The surface of the chief satellite is represented by $99$ uniformly distributed points, which can only be inspected if they lie
within the FOV of the camera, $\alpha_{\text{FOV}} = \pi/3 \ \textrm{rad}$, and are illuminated by the Sun, initially positioned at $\theta_s = 0 \ \text{rad}$ relative to the $x$-axis. The initial relative velocity $\vec{v}_{b/h}(0)$ has a magnitude of $0.3 \ \text{m/s}$ along the same azimuth and elevation directions as the position vector.
\begin{figure}
    \centering
     \input{figures/controlPlot}
    \caption{Trajectories of the thruster commands of the deputy spacecraft.}
    \label{fig:control}
\end{figure}
 The initial position of deputy relative to the chief satellite, $\vec{p}_{b/h}(0)$, is defined as $x = d \cos(\theta_{a}) \cos(\theta_{e})$, $y = d \sin(\theta_{a}) \cos(\theta_{e})$, $z = d \sin(\theta_{e})$, where $d = 50 \text{ m}$, $\theta_{a} = \pi \text{ rad}$, and $\theta_{e}  = 0 \text{ rad}$.
\begin{figure}
    \centering
     \input{figures/dbhiPlot}
    \caption{Trajectories of the actual and estimated distances of the deputy spacecraft relative to features on the chief satellite.}
    \label{fig:dbhi}
\end{figure}
\begin{figure}
    \centering
     \input{figures/dgbPlot}
    \caption{Trajectories of the actual and estimated distance of the deputy spacecraft relative to the key frame.}
    \label{fig:dgb}
\end{figure}
\begin{figure}
    \centering
     \input{figures/dghiPlot}
    \caption{Trajectories of the actual and estimated distances of the key frame relative to features on the chief satellite.}
    \label{fig:dghi}
\end{figure}
\begin{figure}
    \centering
     \input{figures/conditionNumberPlot}
    \caption{Condition number of the regressor $\Sigma_{\mathcal{Y}_{i}}$ for varying orthogonality penalty gains $\gamma_{c}$.}
    \label{fig:condition_number}
\end{figure}

The control objective is to minimize the cost functional in \eqref{eq:coop}, where the state and control penalties are selected to be $Q = \diag([0.1, 0.1, 0.1, 10, 10, 10])$ and $R = 0.1I_{3 \times3}$, respectively, and $Q_{f} = \diag([0.1, 0.1, 0.1, 10, 10, 10])$. The orthogonality penalty is selected as $\gamma_{c} = 1$ and the barrier penalty gain is selected as $\gamma_{\Phi} = 0.1$. Other system parameters are selected as $m = 12 \ \text{kg}$, $n = 0.001027 \ \text{rad/s}$, $r_{d} = 5 \ \text{m}$, $r_{c} = 10 \ \text{m}$, $r_{\max} = 800 \ \text{m}$, $L_{h} = 0.01$, and $a_{\max} = 0.1$. At any given time, only $N = 5$ features within the FOV of the camera are used for distance estimation. The observer learning gain is selected as $k_{\theta} = 1$, with a history stack containing $M = 100$ data points and a time delay $\varsigma = 0.05$. The minimum rank condition in Assumption~\ref{ass:sufficientExcitation} is enforced at $T = 0.25$. Initial distance estimates are selected as $\hat{r}_{b/h_i}(0) = 40 \mathbf{1}_{5\times 1}$, $\hat{r}_{b/k}(0) = 0.01$, and $\hat{r}_{k/h_i}(0) = 40 \mathbf{1}_{5\times 1}$. The fixed distance from the goal to the chief satellite’s origin is $r_{g/h} = 25 \ \text{m}$. The piecewise-constant goal direction $\vec{u}_{g/h}$ is computed using a k-means clustering algorithm and updated according to the dwell-time constraints in Corollary~\ref{cor:DwellTime} with $\delta = 0.1$, ensuring that features remain in the camera’s FOV long enough to satisfy the minimum dwell-time condition. When the $i$-th feature leaves the FOV of the camera, we use its final distance estimate as the initial distance estimate for the subsequent feature entering the FOV. Since the rate of decay of the distance estimation errors $\tilde{\theta}_{i}(t)$ is proportional to the minimum singular value of the regressor matrix $\Sigma_{\mathcal{Y}}$, a singular value maximization algorithm is used to select the time instances $\{t_j\}_{j=1}^{M}$. When the goal location changes, the history stack is purged and replaced with new data points, as described in \cite{SCC.Self.Abudia.ea2022, SCC.Ogri.Bell.ea2023}.

This paper ensures feature illumination when the deputy spacecraft is at the goal location using the k-means clustering algorithm integrated with backward ray tracing to determine whether a point is illuminated or not \cite{SCC.Dunlap.Wijk.ea2023}.  The k-means algorithm serves as a higher-level planner that identifies illuminated points and yields a goal unit vector relative to the chief satellite $\vec{u}_{g/h}(t)$ for $t \in [0, t_{f})$, where $t_{f} \in \mathbb{R}_{\geq 0}$ is the final time for a fixed goal location. At $t = t_{f}$, a new goal location is generated by the k-means algorithm to lie at a fixed distance from the chief satellite along the goal unit vector subject to the dwell time conditions developed in Corollary~\ref{cor:DwellTime}. This goal selection algorithm ensures that the deputy spacecraft only moves towards the nearest cluster of uninspected illuminated points.

\subsection{Results and Discussion}

Figures~\ref{fig:spacecraftTrajectory1000} and \ref{fig:spacecraftTrajectory3000} show that the camera attached to the deputy spacecraft successfully inspects the illuminated features on the chief satellite. The dashed blue line in Figures~\ref{fig:spacecraftTrajectory1000} and \ref{fig:spacecraftTrajectory3000} represents the trajectory of the deputy spacecraft without the barrier function, resulting in a collision with the chief spacecraft. In contrast, the black line depicts the trajectory of the deputy spacecraft under barrier function constraints, ensuring safety by avoiding the collision zone. Figure~\ref{fig:PgbTilde} shows the trajectories of the error between the actual and estimated positions of the goal relative to the deputy spacecraft. The duration for the deputy spacecraft to inspect all points on the chief satellite is mostly determined by the slow orbital motion of the sun relative to the chief satellite, which is assumed to be the only source of illumination within the orbital environment.

The motion of the sun evolves according to \eqref{eq:sunDynamics}, with a separation of 1 AU from the chief satellite. Figures~\ref{fig:dbhi}, \ref{fig:dgb}, and \ref{fig:dghi}, show that the actual and estimated Euclidean distances between the deputy spacecraft, chief satellite, and goal converge to their actual values. The alertness of the MRE-based distance observer enables the controller to accurately guide the deputy spacecraft toward the goal location. Figure~\ref{fig:control} illustrates the trajectory of the thruster forces applied by the deputy during the inspection of the chief satellite.

Figure~\ref{fig:condition_number} highlights the decreasing condition number of the regressor as the orthogonality penalty gains $\gamma_{c}$ increase, consistent with the observations in \cite{SCC.Ogri.Qureshi.ea2024}. This reduction in the condition number indicates enhanced numerical stability of the regressor $\Sigma_{\mathcal{Y}}$, improving the accuracy of Euclidean distance estimation to features on the chief satellite and thereby maximizing feature observability.

\section{Conclusion}\label{section:conclusion}
This paper solves an on-orbit inspection problem in which a deputy spacecraft must safely maneuver within the orbital environment to observe features of a chief satellite, while accounting for illumination and FOV constraints. To satisfy illumination constraints, goal locations are selected using a k-means clustering algorithm, integrated with backward ray tracing to determine illuminated points on the chief satellite. To plan optimal trajectories to guide the deputy spacecraft while executing the inspection task, an MRE-based observer is used to create a local map of the chief satellite using the estimated Euclidean distances to illuminated features on the chief satellite within the FOV of a monocular camera fixed to the deputy spacecraft. To improve the conditioning of the regressor matrix in the MRE-based distance observer \eqref{eq:updateLaw}, a penalty that maximizes the information gained from the camera is incorporated into the cost function of the constrained optimal control problem. An approximate Lagrangian-based controller that maximizes information gain from the camera is then designed to regulate the deputy spacecraft to goal locations selected to facilitate inspection of the chief satellite using a k-means clustering algorithm that serves as a high-level motion planner. The simulation results demonstrate the effectiveness of the developed adaptive controller in carrying out the inspection task.

Future work will focus on extending the developed framework to address more advanced estimation objectives, such as three-dimensional shape reconstruction. Instead of using k-means clustering to obtain goal direction vectors, a hierarchical approach will be employed where a higher-level planner suggests waypoints to improve the $3$-dimensional reconstruction of an object under inspection via a sensor, and a lower-level controller
tracks the waypoints while simultaneously maximizing information gain from the sensor. 
Furthermore, the developed approach will be extended to multiagent on-orbit servicing, where multiple deputies/spacecrafts will be tasked with inspecting points on multiple satellites subject to illumination and FOV constraints. Another extension of this work is to incorporate loop closure, where revisiting previously observed features will be used to improve localization under uncertainty in the velocity of the deputy spacecraft. In this setting, maintaining a memory of past features and recognizing them upon re-encounter will be used to improve estimation performance.

\bibliographystyle{IEEEtaes.bst}
\bibliography{extra,scc,sccmaster,scctemp}
\end{document}

%% file: figures/cameraSchematic.tex
\begin{tikzpicture}[scale=0.525, transform shape]
    \definecolor{corecolor}{RGB}{255, 180, 0} 
    \definecolor{outercolor}{RGB}{255, 204, 0} 
    \definecolor{raycolor}{RGB}{255, 150, 0}   
    \definecolor{darkgreen}{RGB}{0, 100, 0} 

    \tikzstyle{axes}=[->, thick]
    \tikzstyle{labels}=[font=\large]
    \tikzstyle{frame}=[draw, thick]
    \tikzstyle{features}=[fill=black, thick]
    \tikzstyle{dashedlines}=[thin, dashed, opacity=0.6]
    \tikzstyle{vector}=[->, thick, purple]
    \tikzstyle{orbit}=[dotted, thick, gray]
    \tikzstyle{cone}=[draw=blue!50, fill=blue!30, opacity=0.3] 
    \tikzstyle{bluedots}=[draw=blue, fill=red, thick] 

    \coordinate (cameraB) at (-8,11,0);
    \node[below, text=black, labels] at ($(cameraB)+(-0.3,0.5,0.0)$){$\mathcal{O}_{B}$}; 

    \draw[axes, red] (cameraB) -- ++(-1.5, -1.5, 0) node[left, labels]{$\vec{x}_{B}$}; 
    \draw[axes, green] (cameraB) -- ++(0, -2, 0) node[right, labels]{$\vec{y}_{B}$}; 
    \draw[axes, blue] (cameraB) -- ++(2, -1.5, -1) node[right, labels]{$\vec{z}_{B}$}; 

    \coordinate (cameraG) at (-4,7.75,-6);
    \node[below, text=black, labels] at ($(cameraG)+(0,0.1,-1)$){$\mathcal{O}_{G}$}; 

    \draw[axes, red, rotate around x=90, rotate around y=-90] (cameraG) -- ++(1,0,1.5) node[left, labels]{$\vec{x}_{G}$}; 
    \draw[axes, green] (cameraG) -- ++(0,-2,0) node[left, labels]{$\vec{y}_{G}$};
    \draw[axes, blue] (cameraG) -- ++(1.5, -1.5, 0) node[right, labels]{$\vec{z}_{G}$}; 

\begin{scope}
    \fill[gray!70, rounded corners=2pt, opacity=0.9] (3.25,3.2) rectangle (3.75,3.8);

    \fill[blue!40, opacity=0.6] (2.0,3.1) rectangle (3.25,3.9);
    \foreach \x in {2.15, 2.4, 2.65, 2.9, 3.15} {
        \draw[blue!60] (\x,3.1) -- (\x,3.9);
    }
    \foreach \y in {3.25, 3.5, 3.75} {
        \draw[blue!60] (2.0,\y) -- (3.25,\y);
    }

    \fill[blue!40, opacity=0.6] (3.75,3.1) rectangle (5.0,3.9);
    \foreach \x in {3.9, 4.15, 4.4, 4.65, 4.9} {
        \draw[blue!60] (\x,3.1) -- (\x,3.9);
    }
    \foreach \y in {3.25, 3.5, 3.75} {
        \draw[blue!60] (3.75,\y) -- (5.0,\y);
    }

    \draw[fill=gray!40, opacity=0.8] (3.5,3.0) circle [radius=0.15];
    \draw[thick] (3.5,3.2) -- (3.5,3.0);
\end{scope}

    \coordinate (hill) at (3.5,3.5);
    \node[below, text=black, labels] at ($(hill)+(0,-0.65)$){$\mathcal{O}_{H}$}; 
    \draw[axes, red] (hill) -- ++(2,0,0) node[right, labels]{$\vec{x}_{H}$}; 
    \draw[axes, green] (hill) -- ++(0,2,0) node[right, labels]{$\vec{y}_{H}$};
    \draw[axes, blue] (hill) -- ++(0,0,-2) node[right, labels]{$\vec{z}_{H}$};

    \coordinate (keyframe) at (-4.5,7,2.5);
    \node[below, text=black, labels] at ($(keyframe)+(-0.5,-0.15,0.0)$){$\mathcal{O}_{K}$}; 

    \draw[axes, red] (keyframe) -- ++(-1.5, 0, 0) node[left, labels]{$\vec{x}_{K}$}; 
    \draw[axes, green] (keyframe) -- ++(0, -1.5, 0) node[below, labels]{$\vec{y}_{K}$}; 
    \draw[axes, blue] (keyframe) -- ++(0, 0, -1.5) node[right, labels]{$\vec{z}_{K}$}; 

    \path let \p1 = (cameraB), \p2 = (cameraG) in 
          coordinate (splitGB) at ($(\p1)!0.4!(\p2)$);
    \draw[black, ->] (cameraB) -- (splitGB);
    \draw[dashedlines] (splitGB) -- (cameraG);
    \node[labels, above right] at (splitGB) {$\vec{u}_{g/b}(t)$ };

    \path let \p1 = (keyframe), \p2 = (cameraB) in 
          coordinate (splitKB) at ($(\p1)!0.4!(\p2)$);
    \draw[black, ->] (keyframe) -- (splitKB);
    \draw[dashedlines] (splitKB) -- (cameraB);
    \node[labels, above right] at (splitKB) {$\vec{u}_{b/k}(t)$ };

    \path let \p1 = (hill), \p2 = (keyframe) in 
          coordinate (splitKH) at ($(\p1)!0.4!(\p2)$);
    \draw[black, ->] (hill) -- (splitKH);
    \draw[dashedlines] (splitKH) -- (keyframe);
    \node[labels, above right] at (splitKH) {$\vec{u}_{k/h}$ };

    \path let \p1 = (hill), \p2 = (cameraB) in 
          coordinate (splitBH) at ($(\p1)!0.4!(\p2)$);
    \draw[black, ->] (hill) -- (splitBH);
    \draw[dashedlines] (splitBH) -- (cameraB);
    \node[labels, above right] at (splitBH) {$\vec{u}_{b/h}(t)$ };

    \path let \p1 = (hill), \p2 = (cameraG) in 
          coordinate (splitGH) at ($(\p1)!0.4!(\p2)$);
    \draw[black, ->] (hill) -- (splitGH);
    \draw[dashedlines, black] (splitGH) -- (cameraG);
    \node[labels, above right] at (splitGH) {$\vec{u}_{g/h}$ };

    \coordinate (earth) at (-11,1,-7);
    \shade[ball color=blue!50] (earth) circle (1);
    \node[below, labels] at (earth) {Earth};

    \draw[dashedlines] (earth) -- (hill) node[midway, below, labels]{$r_{c}$};

    \draw[orbit] (earth) ++(0,8) arc[start angle=-270, end angle=-360, x radius=11.75, y radius=9] node[midway, above, sloped, gray!60!black, labels] {Chief's Orbit};

    \foreach \i in {0, 45, 90, 135, 180, 225, 270, 315} {
        \draw[bluedots] ($(hill) + (\i:0.5cm)$) circle (2pt); 
    }

    \draw[dashed, magenta] ($(hill) + (225:0.5cm)$) -- ++(-1.3,-0.3);
    \draw[->, magenta] ($(hill) + (225:0.5cm)$) -- ++(-1.3,-0.3)
        node[left, labels] {Inspection point};

    \begin{scope}[shift={(3,10)}]
        \draw[corecolor, thick] (0,0) circle (1cm);
        \foreach \i in {0,30,...,330} {
            \draw[thick, raycolor] (\i:1cm) -- (\i:1.5cm);
        }
        \node at (0,-1.8) {Sun};
    \end{scope}

\end{tikzpicture}

%% file: figures/sunIncidence.tex
\begin{tikzpicture}[rotate=-15, scale=0.6, transform shape]
    \definecolor{corecolor}{RGB}{255, 165, 0} 
    \definecolor{outercolor}{RGB}{255, 160, 0} 
    \definecolor{raycolor}{RGB}{255, 140, 0}   
    \definecolor{shadowcolor}{RGB}{0, 0, 0}    
    \definecolor{bluecolor}{RGB}{0, 0, 255}    
    \definecolor{orangecolor}{RGB}{255, 165, 0} 


\draw[ultra thick, color=gray!50] plot [smooth, tension=1] coordinates {(-2, -0.5) (3, 0.25) (8, -0.5)};

    \draw[->, dashed, color=black] (3, 0.25) -- (3, 3) node[midway, left] {$\vec{n}_{c}$};

    \begin{scope}[shift={(6.5, 5)}]
        \draw[corecolor, thick] (0,0) circle (0.5cm);
        \foreach \i in {0,30,...,330} {
            \draw[thick, raycolor] (\i:0.5cm) -- (\i:1cm);
        }
    \end{scope}

    \begin{scope}[shift={(3, 3.5)}, rotate around={-0.5:(5.5, 3)}]
        \draw[thick, fill=gray!20] (-0.2, -0.2) rectangle (0.2, 0.2);
        \draw[thick, fill=gray!20] (-0.2, -0.2) -- (-0.4, -0.4) -- (0.4, -0.4) -- (0.2, -0.2);
    \end{scope}

    \node[outercolor] at (6.1, 2.5) {$\vec{u}_{s/h}$}; 

\draw[dashed, color=gray!50] (3, 3.5) -- (-1, -0.2); 
\fill[bluecolor] (-1, -0.2) circle (0.1cm); 

\draw[dashed, color=gray!50] (3, 3.5) -- (0, 0.0); 
\fill[bluecolor] (0, 0.0) circle (0.1cm); 

\draw[dashed, color=gray!50] (3, 3.5) -- (1, 0.15); 
\fill[bluecolor] (1, 0.15) circle (0.1cm); 

\draw[dashed, color=gray!50] (3, 3.5) -- (2, 0.25); 
\fill[bluecolor] (2, 0.25) circle (0.1cm); 

\draw[dashed, color=gray!50] (3, 3.5) -- (4, 0.25); 
\fill[orangecolor] (4, 0.25) circle (0.1cm); 

\draw[dashed, color=gray!50] (3, 3.5) -- (5, 0.15); 
\fill[orangecolor] (5, 0.15) circle (0.1cm); 

\draw[dashed, color=gray!50] (3, 3.5) -- (6, 0); 
\fill[orangecolor] (6, 0) circle (0.1cm); 

\draw[dashed, color=gray!50] (3, 3.5) -- (7, -0.2); 
\fill[orangecolor] (7, -0.2) circle (0.1cm); 

\draw[->, thin, dashed, color=outercolor] (6.5, 5) -- (4, 0.25); 
\draw[->, thin, dashed, color=outercolor] (6.5, 5) -- (5, 0.15); 
\draw[->, thin, dashed, color=outercolor] (6.5, 5) -- (6, 0); 
\draw[->, thin, dashed, color=outercolor] (6.5, 5) -- (7, -0.2); 

\node at (0.5, 2) {$\vec{u}_{b/h}$}; 

    \draw[solid, thick, color=red!80] (3, 3) -- (0.5, {1 - (abs(0.5 - 3)/3)});
    \draw[solid, thick, color=red!80] (3, 3) -- (5.5, {1 - (abs(5.5 - 3)/3)});
    \draw[thick, color=red!80] (2.6, 2.45) arc[start angle=220, end angle=320, radius=0.5cm] node[midway, below] {$\alpha_{\text{FOV}}$};

    \draw[solid, thick, color=outercolor] (6.5, 5) -- (3, 0.25);
    \draw[solid, thick, color=outercolor] (6.5, 5) -- (8, -0.5);
    \draw[thick, color=outercolor] (5.75, 4) arc[start angle=230, end angle=300, radius=1cm] node[midway, below] {$\theta_{s}$};

\end{tikzpicture}

%% file: figures/dwellTime.tex
\begin{tikzpicture}
\begin{axis}[
    axis lines=middle,
    xlabel={},
    ylabel={$\|\tilde{\underline{p}}_{g/b}(t)\|$},
    xmin=0, xmax=10,
    ymin=0, ymax=8.5,
    xtick={0,1.0,2.0,3.0,4.5,6,7.35,8.0,10},
    xticklabels={$t_1^a$, $t_2^a$, $t_1^u$, $t_3^a$, $t_2^d$, $t_2^u$, $t_3^d$, $\dots$, $t$},
    ytick={0.5},
    yticklabels={$\epsilon$},
    width=9cm,
    height=6cm,
    tick style={black},
    tick label style={font=\footnotesize},
    label style={font=\footnotesize},
    axis line style={->, thick},
    grid=none,
    clip=false,
    every axis plot/.append style={line width=1pt}
]

\addplot[dashed, gray!80] coordinates {(0,0.5) (10,0.5)};

\addplot[red!80!black, dashed] coordinates {(0,3.5) (2,2.0) (4.5,2.0) (6,2.0) (8.0,2.0) (10,2.0)};
\node[red!80!black, font=\scriptsize, anchor=west] at (axis cs:0.2,3.5) {$h_1$};

\addplot[green!50!black, dashed] coordinates {(1.0,4.25) (3,3.0) (4.5,2.0) (6,1.0) (8.0,1.0) (10,1.0)};
\node[green!50!black, font=\scriptsize, anchor=west] at (axis cs:1.1,4.3) {$h_2$};

\addplot[orange!80!black, dashed] coordinates {(3,4.5) (4.5,3.3) (6,2.0) (8.0,0.5) (10,0.5)};
\node[orange!80!black, font=\scriptsize, anchor=west] at (axis cs:3.1,4.5) {$h_3$};

\addplot[blue, thick, mark=*, mark size=1.5pt] coordinates {
    (0,3.5)
    (2.0,2.0)
    (3.0,2.0)
    (4.5,2.0)
    (6.0,1.0)
    (7.35,1.0)
    (8.00,0.5)
    (10.0,0.5)
};

\foreach \x in {1.0,2,3,4.5,6,7.35,8.0} {
    \addplot[gray!50, dashed] coordinates {(\x,0) (\x,7.5)};
}

\draw[decorate, decoration={brace, amplitude=6pt}, yshift=6pt]
    (axis cs:1.0,4.5) -- (axis cs:4.5,4.5)
    node[midway, yshift=9pt, green!50!black, font=\scriptsize] {dwell time $\Delta t_{2}^{a}$};

\draw[decorate, decoration={brace, amplitude=6pt}, yshift=6pt]
    (axis cs:3.0,5.5) -- (axis cs:7.35,5.5)
    node[midway, yshift=9pt, orange!80!black, font=\scriptsize] {dwell time $\Delta t_{3}^{a}$};
\end{axis}

\end{tikzpicture}

%% file: figures/avoidanceRegionsPlot.tex
\begin{tikzpicture}[scale=0.75, transform shape]
    \tikzstyle{chief} = [circle, fill=blue, minimum size=0.3cm, inner sep=0pt, draw=none]
    \tikzstyle{deputy} = [circle, fill=gray, minimum size=0.3cm, inner sep=0pt, draw=none]
    \tikzstyle{boundary} = [orange, densely dashed, line width=1mm]
    \tikzstyle{collisionZone} = [red, fill=red!20, opacity=0.6]

    \draw[boundary] (0,0) circle (2cm);
    \node[orange!80!black, font=\footnotesize, anchor=west] at (2.3, -1.6) {No-Drift Boundary};
    \draw[->, orange!80!black, line width=0.5mm] (2.3, -1.6) -- (1.5, -1.5);

    \draw[collisionZone] (0,0) circle (1cm);
    \node[red!70!black, font=\footnotesize, anchor=east] at (-1.1, 0.7) {No-Collision Zone};
    \draw[->, red!70!black, line width=0.5mm] (-1.1, 0.7) -- (-0.7, 0.5);

    \node[red!70!black, font=\footnotesize, anchor=south] at (0.6, 0.35) {$r_{\text{min}}$};
    \draw[red!70!black, line width=0.3mm] (0, 0) -- (0.9, 0.52); 

    \node[orange!80!black, font=\footnotesize, anchor=south] at (1.4, -0.4) {$r_{\text{max}}$};
    \draw[orange!80!black, line width=0.3mm] (0, 0) -- (1.9, -0.8); 

    \node[chief] at (0,0) {};
    \node[font=\footnotesize, anchor=north] at (0, -0.25) {Chief Satellite};

    \node[deputy] at (-1.25, -0.75) {};
    \node[font=\footnotesize, anchor=north] at (-1.25, -0.9) {Deputy Spacecraft};
\end{tikzpicture}

%% file: figures/spacecraft3DTrajectory1000.tex

\begin{tikzpicture}
    \begin{axis}[
        view={0}{90}, 
        axis equal,
        xlabel={$x$}, ylabel={$y$}, zlabel={$z$}, 
        legend pos=north east,
        legend style={nodes={scale=0.55, transform shape}},
        height=0.7\columnwidth,
        width=1\columnwidth,
         ymax=25,
          xmax=50,
          zmax=20,
        label style={font=\scriptsize},
        tick label style={font=\scriptsize},
    ]

    \addplot3[
        only marks,
        mark=*,
        mark size=2pt,
        draw=green,
        fill=green,
        fill opacity=0.8
    ] 
    table [col sep=space] {data/inspectedPointsData1000.dat};

    \addplot3[
        only marks,
        mark=*,
        mark size=2pt,
        draw=red,
        fill=red,
        fill opacity=0.8
    ] 
    table [col sep=space] {data/unInspectedPointsData1000.dat};


\draw[fill=white, draw=red!70!black, dashed, thick] 
    (0,0,0) circle (15); 

\addlegendimage{draw=red!70!black, dashed, thick}
\addlegendentry{Collision Zone}

    \addplot3[
        color=black,
        mark=none
    ] 
    table [col sep=space] {data/PData1000.dat};

    \addplot3[
        color=blue,
        dashed,
        mark=none
    ] 
    table [col sep=space] {data/PData1000Unsafe.dat};

    \addplot3[
        only marks,
        mark=square*,
        mark size=3pt,
        draw=black,
        fill=yellow
    ] 
    table [col sep=space] {data/PData_origin1000.dat};

    \addplot3[
        only marks,
        mark=o,
        mark size=3pt,
        draw=orange,
        fill=orange
    ] 
    table [col sep=space] {data/PData_end1000.dat};

    \legend{Inspected, Uninspected, Collision Zone, Safe Trajectory, Unsafe Trajectory, Start, End}
    \end{axis}
\end{tikzpicture}

%% file: figures/spacecraft3DTrajectory3000.tex

\begin{tikzpicture}
    \begin{axis}[
        view={0}{90}, 
        axis equal,
        xlabel={$x$}, ylabel={$y$}, zlabel={$z$}, 
        legend pos=north east,
        legend style={nodes={scale=0.55, transform shape}},
        height=0.7\columnwidth,
        width=1\columnwidth,
          ymax=25,
          xmax=50,
        label style={font=\scriptsize},
        tick label style={font=\scriptsize},
    ]

    \addplot3[
        only marks,
        mark=*,
        mark size=1.5pt,
        draw=green, 
        fill=green,
        fill opacity=0.8
    ] 
    table [col sep=space] {data/inspectedPointsData3000.dat};

    \addplot3[
        only marks,
        mark=*,
        mark size=1.5pt,
        draw=red, 
        fill=red,
        fill opacity=0.8
    ] 
    table [col sep=space] {data/unInspectedPointsData3000.dat};


\draw[fill=white, draw=red!70!black, dashed, thick] 
    (0,0,0) circle (15); 

\addlegendimage{draw=red!70!black, dashed, thick}
\addlegendentry{Collision Zone}

    \addplot3[
        color=black, 
        mark=none 
    ] 
    table [col sep=space] {data/PData3000.dat}; 

    \addplot3[
        color=blue, 
        dashed,
        mark=none 
    ] 
    table [col sep=space] {data/PData3000Unsafe.dat}; 

    \addplot3[
        only marks,
         mark=square*, 
        mark size=3pt,
        draw=black, 
        fill=yellow
    ] 
    table [col sep=space] {data/PData_origin3000.dat}; 

    \addplot3[
        only marks,
        mark=o, 
        mark size=3pt,
        draw=orange, 
        fill=orange
    ] 
    table [col sep=space] {data/PData_end3000.dat}; 

    \addplot3[
        only marks,
        mark=o, 
        mark size=3pt,
        draw=orange, 
        fill=orange
    ] 
    table [col sep=space] {data/PData_end3000UnSafe.dat}; 

    \legend{Inspected, Uninspected, Collision zone, Safe trajectory, Unsafe trajectory, Start, End} 
    \end{axis}
\end{tikzpicture}

%% file: figures/PgbPlot.tex
\begin{tikzpicture}
    \begin{axis}[
        xlabel={$t$ [s]},
        ylabel={$\tilde{\vec{p}}_{g/b}(t)$},
        legend pos=north east,
        legend style={nodes={scale=0.7, transform shape}},
        enlarge y limits=0,
        enlarge x limits=0,
        ymin=-5, ymax=2,
        height=0.6\columnwidth,
        width=1\columnwidth,
        label style={font=\scriptsize},
        tick label style={font=\scriptsize}
    ]

    \pgfplotsinvokeforeach{1,...,3}{
        \addplot+ [mark=none, dashed] table [x index=0, y index=#1] {data/PgbTildeData1000.dat};
    }

    \foreach \x in {50, 100, ..., 950} {
        \edef\temp{\noexpand\draw[dashed, gray] (axis cs:\x,-50) -- (axis cs:\x,120);}
        \temp
    }

    \legend{
        $(\tilde{\vec{p}}_{g/b})_{x}(t)$,
        $(\tilde{\vec{p}}_{g/b})_{y}(t)$,
        $(\tilde{\vec{p}}_{g/b})_{z}(t)$
    }
\end{axis}

\end{tikzpicture}

%% file: figures/controlPlot.tex
\begin{tikzpicture}
    \begin{axis}[
        xlabel={$t$ [s]},
        ylabel={$\mathbf{u}(t)$},
        legend pos= north east,
        legend style={nodes={scale=0.75, transform shape}},
        enlarge y limits=0.1,
        enlarge x limits=0,
        height = 0.6\columnwidth,
        width = 1\columnwidth,
        label style={font=\scriptsize},
        tick label style={font=\scriptsize}
    ]
    \pgfplotsinvokeforeach{1,...,3}{
        \addplot+ [ mark=none] table [x index=0, y index=#1] {data/uData1000.dat};
    }
    \foreach \x in {50, 100, ..., 950} {
        \edef\temp{\noexpand\draw[dashed, gray] (axis cs:\x,-40) -- (axis cs:\x,40);}
        \temp
    }
    \legend{$(F(t))_{x}$, $(F(t))_{y}$,  $(F(t))_{z}$} 
    \end{axis}
\end{tikzpicture}

%% file: figures/dbhiPlot.tex
\begin{tikzpicture}
    \begin{axis}[
        xlabel={$t$ [s]},
        ylabel={$r_{b/h}(t)$},
        legend pos= north east,
        legend style={nodes={scale=0.75, transform shape}},
        enlarge y limits=0.1,
        enlarge x limits=0,
        height = 0.6\columnwidth,
        width = 1\columnwidth,
        label style={font=\scriptsize},
        tick label style={font=\scriptsize}
    ]
    \pgfplotsinvokeforeach{1,...,5}{
        \addplot+ [color=color#1, mark=none] table [x index=0, y index=#1] {data/dbhiData1000.dat};
        \addlegendentry{$r_{b/h_{#1}}(t)$}
    }
    \pgfplotsinvokeforeach{1,...,5}{
        \addplot+ [dashed, color=color#1, mark=none] table [x index=0, y index=#1] {data/dbhiHatData1000.dat};
        \addlegendentry{$\hat{r}_{b/h_{#1}}(t)$}
    }
    \foreach \x in {50, 100, ..., 950} {
        \edef\temp{\noexpand\draw[dashed, gray] (axis cs:\x,0) -- (axis cs:\x,55);}
        \temp
    }
    \end{axis}
\end{tikzpicture}

%% file: figures/dgbPlot.tex
\begin{tikzpicture}
    \begin{axis}[
        xlabel={$t$ [s]},
        ylabel={$r_{b/k}(t)$},
        legend pos=  north east,
        legend style={nodes={scale=0.75, transform shape}},
        enlarge y limits=0.1,
        enlarge x limits=0,
        height = 0.6\columnwidth,
        width = 1\columnwidth,
        label style={font=\scriptsize},
        tick label style={font=\scriptsize},
    ]
    \pgfplotsinvokeforeach{1}{
        \addplot+ [ mark=none, color=blue] table [x index=0, y index=#1] {data/dkbData1000.dat};
    }
    \pgfplotsinvokeforeach{1}{
        \addplot+ [ dashed, mark=none, color=red] table [x index=0, y index=#1] {data/dkbHatData1000.dat};
    }
    \foreach \x in {50, 100, ..., 950} {
        \edef\temp{\noexpand\draw[dashed, gray] (axis cs:\x,-5) -- (axis cs:\x,60);}
        \temp
    }

    \legend{$r_{b/k}(t)$, $\hat{r}_{b/k}(t)$}
    
    \end{axis}

    
    
    
\end{tikzpicture}

%% file: figures/dghiPlot.tex
\begin{tikzpicture}
    \begin{axis}[
        xlabel={$t$ [s]},
        ylabel={$r_{k/h}(t)$},
        legend pos= north east,
        legend style={nodes={scale=0.75, transform shape}},
        enlarge y limits=0.1,
        enlarge x limits=0,
        height = 0.6\columnwidth,
        width = 1\columnwidth,
        label style={font=\scriptsize},
        tick label style={font=\scriptsize}
    ]
    \pgfplotsinvokeforeach{1,...,5}{
        \addplot+ [ color=color#1, mark=none] table [x index=0, y index=#1] {data/dkhiData1000.dat};
        \addlegendentry{$r_{k/h_{#1}}(t)$}
    }
    \pgfplotsinvokeforeach{1,...,5}{
        \addplot+ [ color=color#1, mark=none] table [x index=0, y index=#1] {data/dkhiHatData1000.dat};
        \addlegendentry{$\hat{r}_{k/h_{#1}}(t)$}
    }
    \foreach \x in {50, 100, ..., 950} {
        \edef\temp{\noexpand\draw[dashed, gray] (axis cs:\x,0) -- (axis cs:\x,60);}
        \temp
    }
    \end{axis}

    
    
    
\end{tikzpicture}

%% file: figures/conditionNumberPlot.tex
\begin{tikzpicture}[new spy style]
    \begin{axis}[
       xlabel={$t$ [s]}, 
        ylabel={$Cond\left(\Sigma_{\mathcal{Y}}(t)\right)$}, 
        legend pos= north east,
        legend style={nodes={scale=0.55, transform shape}},
        ymax=50, 
        width=1\columnwidth,
        height=0.7\columnwidth,
    ]

    \addplot[color=color1,  mark=none] table [x index=0, y index=1, col sep=space] {data/allConditionNumbersData0.dat};
    \addlegendentry{$\gamma_c = 0$}

    \addplot[color=color2,  mark=none] table [x index=0, y index=1, col sep=space] {data/allConditionNumbersData1.dat};
    \addlegendentry{$\gamma_c = 5$}

    \addplot[color=color3,  mark=none] table [x index=0, y index=1, col sep=space] {data/allConditionNumbersData2.dat};
    \addlegendentry{$\gamma_c = 10$}

    \addplot[color=color4,  mark=none] table [x index=0, y index=1, col sep=space] {data/allConditionNumbersData3.dat};
    \addlegendentry{$\gamma_c = 15$}

    \addplot[color=color5,  mark=none] table [x index=0, y index=1, col sep=space] {data/allConditionNumbersData4.dat};
    \addlegendentry{$\gamma_c = 20$}

   \coordinate (spypoint) at (axis cs:6,5); 
    \coordinate (magnifyglass) at (axis cs:5,30); 

    \end{axis}

    \spy [gray, width=3.5cm ,height=1.5cm] on (spypoint) in node[fill=white] at (magnifyglass);
\end{tikzpicture}